\RequirePackage{fix-cm}
\documentclass[smallextended]{svjour3}       % onecolumn (second format)
\smartqed  % flush right qed marks, e.g. at end of proof
\usepackage{graphicx}
\usepackage{natbib}

\usepackage{fancyhdr}
\usepackage{amssymb}
\usepackage{amsmath}
\usepackage[font=scriptsize,labelfont=bf]{caption}
\usepackage{sidecap}
\usepackage{booktabs}
\usepackage[latin1]{inputenc}
\usepackage{float}
\usepackage{subfig}
\usepackage{url}
\usepackage{float}
\usepackage[figuresright]{rotating}
\usepackage{enumerate}

\newtheorem{thm}{Theorem}[section]
\newtheorem{prop}[thm]{Proposition}
\newtheorem{lem}[thm]{Lemma}

\numberwithin{equation}{section}
\begin{document}

\title{Modeling of Contact Tracing in Epidemic Populations Structured by Disease Age%\thanks{Grants or other notes
}


\author{Xi Huo}


\institute{Xi Huo \at
              Department of Mathematics, Vanderbilt University, Nashville, TN 37240, USA \\
              \email{xi.huo@vanderbilt.edu} }

\date{Received: date / Accepted: date}

\maketitle

\begin{abstract}
We consider an age-structured epidemic model with two basic public health interventions: (i) identifying and isolating symptomatic cases, and (ii) tracing and quarantine of the contacts of identified infectives. The dynamics of the infected population are modeled by a nonlinear infection-age-dependent partial differential equation, which is coupled with an ordinary differential equation that describes the dynamics of the susceptible population. Theoretical results about global existence and uniqueness of positive solutions are proved. We also present two practical applications of our model: (1) we assess public health guidelines about emergency preparedness and response in the event of a smallpox bioterrorist attack; (2) we simulate the 2003 SARS outbreak in Taiwan and estimate the number of cases avoided by contact tracing. Our model can be applied as a rational basis for decision makers to guide interventions and deploy public health resources in future epidemics. 
\end{abstract}

\section{Introduction}
\label{intro}
Our aim is to develop a model to assess the effectiveness of two public health interventions in controlling epidemic outbreaks: (i) identifying and isolating symptomatic cases, and (ii) tracing of their contacts, followed by isolation, quarantine, or vaccination. Our model is applicable to general epidemics for which quarantine or vaccination are available as control measures. In many cases, there are two alternatives for such controls, namely targeted control or mass control. Isolation of symptomatic cases is important in controlling infectious diseases, but also important may be the vaccination and quarantine of traced contacts of known infectives. Contact tracing is especially important when there is a lack of rapid diagnostic methods, as was in the case of SARS (\citealt{Glasser}). \\

Ordinary differential equation models were used in modeling SARS, and in particular to investigate the impact of quarantining asymptomatic infectives (\citealt{Hethcote, WangRuan, Hsieh3, Brauer, Nishiura, Fraser, Day, Hsieh, Pauline, Feng3, Feng, Feng2}). A thorough review of many of these works has been provided in (\citealt{Galvani}). The article points out that the nonlinearity of the rate of quarantining undiagnosed cases is required to be taken into account. In this paper, we use a partial differential equation model with a variable of disease age (or age since infection), with nonlinear rates of contact tracing infectives and quarantining susceptibles dependent on the rate of identifying symptomatic cases.\\

In the smallpox application, our model applies to the ring vaccination strategy which is believed to be the reason for smallpox eradication. Various methods have been developed to evaluate public health control strategies for smallpox, such as stochastic models (\citealt{Muller, Meltzer, Halloran, Eichner, Kretzschmar, Vidondo}), and ordinary differential equation models (\citealt{Kaplan, KaplanPNAS, CCC}). \\

Although some of the previous work include the varying levels of transmission ability and symptom scores in different disease stages, there is less work about smallpox control that takes continuous disease age into consideration. Webb \emph{et~al.}\ apply age structured epidemic models to investigate isolation strategy and school closings in the spread of H1N1 (\citealt{webb2010}). Inaba \emph{et~al.}\ develop a series of multistate class age structured epidemic systems with isolation rate as the only intervention (\citealt{Inaba}). Fraser \emph{et~al.}\ establish an infection age-structured model that estimates the effectiveness of isolation and contact tracing in the control of epidemic diseases with a formulation different from ours (\citealt{Fraser}). Our model is aimed to take into consideration several key features about disease transmission and public health interventions at the same time: (i) continuous infection age; (ii) infection-age-dependent case isolation rate; (iii) contact tracing/quarantine/vaccination rates that depend on diagnosis rate of symptomatic cases; (iv) variation of susceptible population due to infection and contact tracing/quarantine/vaccination.\\

We present the general model in Section \ref{sec:2}. In Section \ref{sec:3}, the preliminary results about the solutions to the general problem are stated. In Section \ref{sec:4}, we illustrate more analysis results for the specific problem, and we provide all of the proofs in Appendix. Our model application to smallpox, as well as practical interpretations of the model parameters, are demonstrated in Section \ref{sec:5}. Then we investigate the effectiveness of contact tracing strategies which were implemented in control of 2003 SARS in Taiwan as the second application in Section \ref{sec:6}.

\section{A Logistic Age-Dependent Epidemic Model}
\label{sec:2}
First of all, we clarify the definitions of the three intervention strategies considered in the model: (1) isolation is the process by which infected people (all of whom are symptomatic) are prevented from infecting susceptible ones; (2) contact tracing is the process of identifying people who may have been infected by exposure to (or contact with) an infectious person; (3) quarantine is the process of isolating these people, called contacts (none of whom is yet symptomatic, and many are not even infected)\footnote{Thanks to J. Glasser at Centers for Disease Control and Prevention for kindly providing the definitions.}. The conduct of quarantine can be divided into two types: (i) quarantine close contacts of identified cases, and (ii) quarantine large groups of people (such as residents in residential complexes, workers in a workplace, students in schools, \emph{etc}). Our model focuses on the (i) type quarantine, which is conducted as a consequence of tracing close contacts of infected individuals.\\
Before introducing the main model, we introduce the notations as follows:
\begin{enumerate}[({N.}1)]
\item
For $0 < \emph{M} \leqslant \infty $, let ${L^1}: = {L^1}\left( {\left[ {0,\emph{M}} \right];{\mathbb{R}}} \right)$, $L_ + ^1: = {L^1}\left( {\left[ {0,\emph{M}} \right];{{\mathbb {R}}_ + }} \right)$ which is the positive cone in ${L^1}$. $\emph{M}$ denotes the maximum disease age in the model.
\item
For $0 < T \leqslant \infty $, denote ${C_T}: = C\left( {\left[ {0,T} \right];{L^1}} \right)$ with the supremum norm: ${\left\| l \right\|_{{C_T}}}: = \mathop {\sup }\limits_{0 \leqslant t \leqslant T} {\left\| {l\left( t \right)} \right\|_{{L^1}}}$, for $l \in {C_T}$. Let ${C_{T, + }}: = C\left( {\left[ {0,T} \right];L_ + ^1} \right)$ which is the positive cone in ${C_T}$.
\end{enumerate}
The basic assumptions are as follows:
\begin{enumerate}[({A.}1)]
\item
Let $\mathcal{T} \in {\left( {{L^1}} \right)^*}$ with norm ${\left\| \mathcal{T} \right\|_\infty }$, and we assume $\mathcal{T}\left( {L_ + ^1} \right) \subseteq {\mathbb{R}_ + }$.
\item
$\mathcal{B} ,\mathcal{Q} :{L^1} \to \mathbb{R}$ be globally Lipschitz continuous functions with Lipschitz constants $\left| \mathcal{B} \right|$ and $\left| \mathcal{Q} \right|$. Moreover, we assume $\mathcal{B}\left( {L_ + ^1} \right) \subseteq {\mathbb{R}_ + }$ and $\mathcal{Q}\left( {L_ + ^1} \right) \subseteq {\mathbb{R}_ + }$.
\item
$\mathcal{B}\left( 0 \right) = 0$, $\mathcal{Q}\left( 0 \right) = 0$.
\end{enumerate}

For $t \geqslant 0$, $a \in \left[ {0,\emph{M}} \right]$, the formal model is:

\begin{flalign}\label{whole model}
\begin{split}
&\frac{\partial }{{\partial t}}i\left( {a,t} \right){\text{  +  }}\frac{\partial }{{\partial a}}i\left( {a,t} \right) = \underbrace { - \mu \left( a \right)i\left( {a,t} \right)}_{\substack{\text{isolation of symptomatic}\\\text{individuals with infection}\\\text{age $a$}}} \underbrace {- {\mathcal{T}}\left( {i\left( {\cdot,t} \right)} \right)i\left( {a,t} \right)}_{\substack{\text{tracing of contacts}\\\text{with infection age $a$}}}\\
&\frac{d}{{dt}}S\left( t \right) = \underbrace {- \mathcal{B} \left( {i\left( { \cdot ,t} \right)} \right)S\left( t \right)}_{\substack {\text{infection of susceptibles}}} \quad \underbrace{- \mathcal{Q}\left( {i\left( { \cdot ,t} \right)} \right)S\left( t \right)}_{\substack{\text{quarantine of contacts}\\\text{that are susceptible}}}\\
&\underbrace{i\left( {0,t} \right)}_{\substack{\text{infectives with}\\\text {infection age $0$}}} = \underbrace{\mathcal{B} \left( {i\left( { \cdot ,t} \right)} \right)S\left( t \right)}_{\substack{\text{rate of new infections}}}\\
&i(a,0) = {i_0}(a) \in {L^1}\left[ {0,\emph{M}} \right]\\
&S(0) = {S_0} \in {\mathbb{R}_+}
\end{split}
\end{flalign}
where $i(a,t)$ is the infected population density at infection age $a$ at time $t$, and $S(t)$ is the susceptible population at time $t$, $\mu (a)$ is the rate of isolating symptomatic cases those are at disease age $a$. If we denote $i\left( { \cdot ,t}\right)$ as the infected population density function at time $t$, then $\mathcal{B}\left( {i\left( { \cdot ,t} \right)} \right)$ represents the infection transmission rate, $\mathcal{T}\left( {i\left( { \cdot ,t} \right)} \right)$ represents the isolation rate of infected individuals due to contact tracing at time $t$,  and $\mathcal{Q}\left( {i\left( { \cdot ,t} \right)} \right)$ represents the quarantine rate of susceptible contacts as the consequence of contact tracing.\\

Solving for $S\left( t \right)$ from the second equation in \eqref{whole model}, we can simplify the problem into an age-dependent population dynamics model for $i\left( {a,t} \right)$:
\begin{flalign}\label{reduced model}
\begin{split}
&\frac{\partial }{{\partial t}}i\left( {a,t} \right){\text{ + }}\frac{\partial }{{\partial a}}i\left( {a,t} \right) =  - \mu \left( a \right)i\left( {a,t} \right) - \mathcal{T}\left( {i\left( { \cdot ,t} \right)} \right)i\left( {a,t} \right)\\
&i\left( {0,t} \right) = {S_0} \mathcal{B} \left( {i\left( { \cdot ,t} \right)} \right){e^{ - \int_0^t {\mathcal{B} \left( {i\left( { \cdot ,s} \right)} \right) + \mathcal{Q}\left( {i\left( { \cdot ,s} \right)} \right)ds} }}\\
&i(a,0) = {i_0}(a)
\end{split}
\end{flalign}

In the following context, we denote $i\left( t \right)\!\left( a \right): = i\left( {a,t} \right)$; then $i \in {C_T}$ means $i\!\left( t \right) \in {L^1}$, for $t \in \left[ {0,T} \right]$. We also refer the solutions of the age-dependent problem as $l\left( t \right)\!\left( a \right): = i\left( {a,t} \right)$, where $l \in {C_T}$. Next we will generalize the problem to a formulation of age-dependent population dynamics. We introduce the aging and birth functions: 
\begin{enumerate}[({}1)]
\item
Let $G:{L^1} \to {L^1}$ be the aging function.
\item
For $0 < T \leqslant \infty $, let $F: {C_T} = C\left( {\left[ {0,T } \right];{L^1}} \right) \to C\left( {\left[ {0,T } \right];\mathbb{R}} \right)$ be the birth function. 
\end{enumerate}

Let $0 < T \leqslant \infty $, let $t \in \left[ {0,T} \right]$ and $l \in {C_T}$, the general age-dependent problem is as follows:
\begin{flalign}\label{pre ADP}
\begin{split}
&\frac{\partial }{{\partial t}}l\left( t \right)\!\left( a \right) + \frac{\partial }{{\partial a}}l\left( t \right)\!\left( a \right) = G\!\left( {l\left( t \right)} \right)\!\left( a \right),\,a.e.\ a \in \left[ {0,\emph{M}} \right]\\
&l\left( t \right)\!\left( 0 \right) = \left( {F\left( l \right)} \right)\!\left( t \right)\\
&l\left( 0 \right)\!\left( a \right) = \phi \left( a \right),\,a.e.\ a \in \left[ {0,\emph{M}} \right] 
\end{split}
\end{flalign}

However, the equation system in \eqref{pre ADP} is not well defined for solutions that are not continuously differentiable with respect to both variables. We are thus led to the following formulation of age-dependent population dynamics as in (\citealt{webbbook}): let $0 < T \leqslant \infty $, and let $l \in {C_T}$ satisfy:
\begin{equation}\label{ADP}
\begin{split}
&\mathop {\lim }\limits_{h \to {0^ + }} \int_0^\emph{M} {\left| {{h^{ - 1}}\left[ {l\left( {t + h} \right)\!\left( {a + h} \right) - l\left( t \right)\!\left( a \right)} \right] - G\!\left( {l\left( t \right)} \right)\!\left( a \right)} \right|da}  = 0\\
&\mathop {\lim }\limits_{h \to {0^ + }} \int_0^h {\left| {l\left( {t + h} \right)\!\left( a \right) - \left( {F\left( l \right)} \right)\!\left( t \right)} \right|da}  = 0\\
&l\left( 0 \right)\!\left( a \right) = \phi \left( a \right),\; a.e.\ a \in \left[ {0,\emph{M}} \right]
\end{split}
\end{equation}
where we let $l\left( {t + h} \right)\!\left( {a + h} \right) = 0$ if $a + h > \emph{M}$.
%%%%%%%%%%%%%%%%%%%%%%%%%%%%%%%%%%%%%%%%%%%%%%%%%%%%%%%%%%%%%%%%%%%%%%%%%%%%%%%%%%%%%%%%%%%%%%%%%%%%%%%%%%%%%%%%%%%%%%%%%%%%%%%%%%%%%%%%%%%%%%%%%%%%%%%%%%%%%%%%%%%%%%%%%%%%%%%%%%%%%%%%%%%%%%%%%%%%%%%%%%%%%%%%%%%%%%%%%%%%%%%%%%%%%%%%%%%%%%%%%%%%%%%%%%%%%%%%%%%%%%%%%%%%%%%%%%%%%%%%%%%%%%%%%%%%%%%%%%%%%%%%%%%%%%%%%%%%%%%%%%%%%%%%%%%%%%%%%%%%%%%%%%%%%%%%%%%%%%%%%%%%%%%%%%%%%%%%%%%%%%%%%%%%%%%%%%%%%%%%%%%%%%%%%%%%%%%%%%%%%%%%%%%%%%%%%%%%%%%%%%%%%%%%%%%%%%%%%%%%%%%%%%%%%%%%%%%%%%%%%%%%%%%%%%%%%%%%%%%%%%%%%%%%%%%%%%
\section{Preliminary Results}
\label{sec:3}
In this section, we will present results about local existence and uniqueness of the solutions to the age-dependent problem \eqref{ADP} with the following assumptions on the aging and birth functions:
\begin{enumerate}[({H.}1)]
\item
$G:{L^1} \to {L^1}$, there is an increasing function ${c_1}:\left[ {0,\infty } \right) \to \left[ {0,\infty } \right)$ such that ${\left\| {G\left( {{\phi _1}} \right) - G\left( {{\phi _2}} \right)} \right\|_{{L^1}}} \le {c_1}\left( r \right){\left\| {{\phi _1} - {\phi _2}} \right\|_{{L^1}}}$
for all ${\phi _1},{\phi _2} \in {L^1}$ such that ${\left\| {{\phi _1}} \right\|_{{L^1}}},{\left\| {{\phi _2}} \right\|_{{L^1}}} \le r$.
\item
There is a function ${c_2}:\left[ {0,\infty } \right) \times \left[ {0,\infty } \right) \to \left[ {0,\infty } \right)$, which is increasing and continuous w.r.t.\ both variables. Then for all $T>0$, $F: {C_T} \to C\left( {\left[ {0,T} \right];\mathbb{R}} \right)$, for any $0 \le t \le T$ and $r>0$, we have
$$\left| {\left( {F\left( {{\phi _1}} \right)} \right)\!\left( t \right) - \left( {F\left( {{\phi _2}} \right)} \right)\!\left( t \right)} \right| \leqslant {c_2}\left( {r,t} \right)\mathop {\sup }\limits_{0 \leqslant s \leqslant t} {\left\| {{\phi _1}\left( s \right) - {\phi _2}\left( s \right)} \right\|_{{L^1}}}$$
for all ${\phi _1}$, ${\phi _2} \in {C_T}$ such that ${\left\| {{\phi _1}} \right\|_{{C_T}}},{\left\| {{\phi _2}} \right\|_{{C_T}}} \le r$.
\end{enumerate}
We state theorems about local existence and uniqueness of the solutions below. The proofs (they can be found in the Appendix) are different from those in (\citealt{webbbook}), since our assumption of the birth function $F$ is different.
\begin{thm}\label{local existence and uniqueness}
Let (H.1) and (H.2) hold and let $\phi  \in {L^1}$. There exists $T>0$ and $l \in {C_T}$ such that $l$ is a solution of \eqref{ADP} on $[0,T]$. Furthermore, there is a unique solution of \eqref{ADP} on $[0,T]$.
\end{thm}
We introduce the definition of maximal interval of existence as in (\citealt{webbbook}):
\begin{definition}
Let $\phi  \in {L^1}$. Denote $\left[ {0,{T_\phi }} \right)$ as the maximal interval of existence of the solution of \eqref{ADP}, is the maximal interval with the property that if $0 < T < {T_\phi }$, there exist $l \in {C_T}$ such that $l$ is a solution of \eqref{ADP} on $\left[ {0,T} \right]$.
\end{definition}
With additional assumptions as stated below, we will prove the positivity of the solutions.
\begin{description}
\item[$(H.3)$]
$F\left( {C_{T, + }} \right) \subseteq C\left( {\left[ {0,T} \right];{\mathbb{R}_ + }} \right) $
\item[$(H.4)$]
There is an increasing function ${c_3}:\left[ {0,\infty } \right) \to \left[ {0,\infty } \right)$ such that if $r>0$ and $\phi  \in L_ + ^1$ with ${\left\| \phi  \right\|_{{L^1}}} \le r$, then $G\left( \phi  \right) + {c_3}\left( r \right)\phi  \in L_ + ^1$.
\end{description}

\begin{thm}\label{existence and uniqueness of positive solution}
Let (H.1)-(H.4) hold and let $\phi  \in {L_+^1}$. The solution $l$ of \eqref{ADP} on $\left[ {0,{T_\phi }} \right)$, has the property that $l\left( t \right) \in L_ + ^1$ for $0 \leqslant t < {T_\phi }$.
\end{thm}
Furthermore, with one more restriction on the aging and birth functions, the positive solution exists globally.
\begin{thm}\label{global solution}
Let (H.1)-(H.4) hold and let $\phi  \in L_ + ^1$, let $l$ be the solution of \eqref{ADP} on $\left[ {0,{T_\phi }} \right)$, and let there exist $\omega  \in \mathbb{R}$ such that for $0 \le t < {T_\phi }$, $F$ and $G$ satisfy the following inequality:
\begin{equation}\label{global condition}
\left( {F\left( l \right)} \right)\!\left( t \right) + \int_0^\emph{M} {G\left( {l\left( t \right)} \right)\!\left( a \right)da}  \leqslant \omega \int_0^\emph{M} {l\left( t \right)\!\left( a \right)da} \tag{H.5}
\end{equation}
Then ${T_\phi } = \infty $ and ${\left\| {l\left( {t} \right)} \right\|_{{L^1}}} \le {e^{\omega t}}{\left\| \phi  \right\|_{{L^1}}}$.
\end{thm}
%%%%%%%%%%%%%%%%%%%%%%%%%%%%%%%%%%%%%%%%%%%%%%%%%%%%%%%%%%%%%%%%%%%%%%%%%%%%%%%%%%%%%%%%%%%%%%%%%%%%%%%%%%%%%%%%%%%%%%%%%%%%%%%%%%%%%%%%%%%%%%%%%%%%%%%%%%%%%%%%%%%%%%%%%%%%%%%%%%%%%%%%%%%%%%%%%%%%%%%%%%%%%%%%%%%%%%%%%%%%%%%%%%%%%%%%%%%%%%%%%%%%%%%%%%%%%%%%%%%%%%%%%%%%%%%%%%%%%%%%%%%%%%%%%%%%%%%%%%%%%%%%%%%%%%%%%%%%%%%%%%%%%%%%%%%%%%%%%%%%%%%%%%%%%%%%%%%%%%%%%%%%%%%%%%%%%%%%%%%%%%%%%
\section{Basic Theory of the Logistic Age-Dependent Epidemic Model}
\label{sec:4}
We will continue investigating the solutions of the specific age-dependent problem \eqref{reduced model} in the sense of \eqref{ADP}. First, specify the birth and aging functions:
\begin{enumerate}[({P.}1)]
\item
The aging function $G:{L^1} \to {L^1}$ is, for $\phi  \in {L^1}$,
\\$G\left( \phi  \right)\!\left( a \right) =  - \mu \left( a \right)\phi \left( a \right) - \mathcal{T}\left( \phi  \right)\phi \left( a \right)$.
\item
The birth function $F:{C_T} \to C\left( {\left[ {0,\infty} \right];\mathbb{R}} \right)$ is, for $l \in {C_T}$,
\\$\left( {F\left( l  \right)} \right)\!\left( t \right): = {S_0}\mathcal{B} \left( {l \left( t \right)} \right){e^{ - \int_0^t {\mathcal{B} \left( {l \left( s \right)} \right) + \mathcal{Q}\left( {l \left( s \right)} \right)ds} }}$.
\end{enumerate}
where $\mu$ and ${S_0}$ are as in \eqref{whole model}, $\mathcal{T}$, $\mathcal{B}$, and $\mathcal{Q}$ are as in (A.1)-(A.3).
\begin{thm}\label{existence and uniqueness of main problem}
Let (A.1), (A.2) and (A.3) hold, let $\mu  \in L_ + ^\infty \left[ {0,\emph{M}} \right]$, ${S_0} > 0$, and $\phi \in L_ + ^1\left[ {0,\emph{M}} \right]$. There is a function $l \in C\left( {\left[ {0,\infty } \right);L_ + ^1} \right)$ such that $l$ is the unique global solution of \eqref{ADP} with the aging function $G$ and birth function $F$ in (P.1) and (P.2).
\end{thm}
For computational convenience and the proof of the asymptotic behavior, we introduce a solution formula in the following two theorems.
\begin{prop}\label{existence of related problem}
Let (A.1), (A.2) and (A.3) hold, let $\mu  \in L_ + ^\infty \left[ {0,\emph{M}} \right]$, ${S_0} > 0$, and $\phi \in L_ + ^1$. There exists $u \in C\left( {\left[ {0,\infty} \right];L_ + ^1} \right)$ such that $u$ satisfies:
\begin{equation}\label{u formula}
  u\!\left( t \right)\!\left( a \right) = \left\{ \begin{matrix} 
  \left( {\mathcal{H}\left( u \right)} \right)\!\left( {t - a} \right){e^{ - \int_0^a {\mu \left( b \right)db} }},\;a < t \hfill \cr 
  \phi\left( {a - t} \right){e^{ - \int_{a - t}^a {\mu \left( b \right)db} }},\;a \ge t \hfill \cr 
 \end{matrix}  \right.
\end{equation}
where 
\begin{flalign*}
\left( {\mathcal{H}\left( u \right)} \right)\!\left( t \right)= & {S_0}\mathcal{B} \left( {\frac{{u\left( t \right)}}{{1 + \int_0^t {\mathcal{T}\left( {u\left( s \right)} \right)ds} }}} \right) \times \left( {1 + \int_0^t {\mathcal{T}\left( {u\left( s \right)} \right)ds} } \right) \times \\
& {e^{ - \int_0^t {\mathcal{B} \left( {\frac{{u\left( s \right)}}{{1 + \int_0^s {\mathcal{T}\left( {u\left( \tau  \right)} \right)d\tau } }}} \right) + \mathcal{Q}\left( {\frac{{u\left( s \right)}}{{1 + \int_0^s {\mathcal{T}\left( {u\left( \tau  \right)} \right)d\tau } }}} \right)ds} }}
\end{flalign*}
Moreover, $u$ is a solution of \eqref{ADP} with birth function $\mathcal{H}$ and aging function $\mathcal{P}:{L^1} \to {L^1}$, where $\mathcal{P}\left( l \right)\!\left( a \right) =  - \mu \left( a \right)l\left( a \right)$ for $l \in {L^1}$.
\end{prop}
\begin{thm}\label{solution in form of u}
Let (A.1), (A.2), and (A.3) hold, let $\mu  \in L_ + ^\infty \left[ {0,\emph{M}} \right]$, ${S_0} > 0$, and $\phi \in L_ + ^1$. Let $u \in C\left( {\left[ {0,\infty } \right);L_ + ^1} \right)$ be the solution to the integral equation \eqref{u formula}. Then
\begin{equation}\label{solution formula}
l\left( t \right)\!\left( a \right) = \frac{u\left( t \right)\!\left( a \right)}{{1 + \int\limits_0^t {\mathcal{T}(u(\tau ))d\tau } }}
\end{equation}
gives the unique global solution $l \in C\left( {\left[ {0,\infty } \right);L_ + ^1} \right)$to problem \eqref{reduced model}. 
\end{thm}
By formula \eqref{solution formula}, we will be able to investigate the asymptotic behaviour of the solution to \eqref{reduced model} as in the next theorem. Furthermore, \eqref{solution formula} also provides a starting point for our simulations.
\begin{thm}\label{asymptotic}
Let (A.1), (A.2) and (A.3) hold, let $\mu  \in L_ + ^\infty \left[ {0,\emph{M}} \right]$, ${S_0} > 0$, and $\phi \in L_ + ^1$. Assume that there is an ${a_0} \in \left[ {0,\emph{M}} \right)$ such that $\mu \left( a \right) > {\mu_0 }  > 0$ for all $a \in \left[ {{a_0},\emph{M}} \right]$. Then, for the unique solution of \eqref{whole model} in the sense of \eqref{ADP}, $\mathop {\lim }\limits_{t \to \infty } S\left( t \right) = {S_\infty } > 0$, $\mathop {\lim }\limits_{t \to \infty } I\left( t \right)= \mathop {\lim }\limits_{t \to \infty } \int_0^\emph{M} {i\left( {a,t} \right)da}  = 0$.
\end{thm}
%%%%%%%%%%%%%%%%%%%%%%%%%%%%%%%%%%%%%%%%%%%%%%%%%%%%%%%%%%%%%%%%%%%%%%%%%%%%%%%%%%%%%%%%%%%%%%%%%%%%%%%%%%%%%%%%%%%%%%%%%%%%%%%%%%%%%%%%%%%%%%%%%%%%%%%%%%%%%%%%%%%%%%%%%%%%%%%%%%%%%%%%%%%%%%%%%%%%%%%%%%%%%%%%%%%%%%%%%%%%%%%%%%%%%%%%%%%%%%%%%%%%%%%%%%%%%%%%%%%%%%%%%%%%%%%%%%%%%%%%%%%%%%%%%%%%%%%%%%%%%%%%%%%%%%%%%%%%%%%%%%%%%%%%%%%%%%%%%%%%%%%%%%%%%%%%%%%%%%%%%%%%%%%%%%%%%%%%%%%%%%%%%
\section{Application I: Smallpox}
\label{sec:5}
Smallpox was eradicated in 1979, but fears of bioterrorist attacks by deliberately releasing the variola virus have been taken into consideration according to federal and academic observations ever since the terrorist attacks of September 11, 2001. Although the two government laboratories in the United States and Russia are the only known places that keep the viral samples, the possibility of other sources cannot be ruled out (\citealt{CIA}). Public health authorities have detailed plans for emergency preparedness and response to a smallpox outbreak (\citealt{CDCresponse}); on the other hand, the proper amount of vaccine and treatment medicine that should be stockpiled is still controversial (\citealt{NYT}).\\

%Whenever we try to answer questions like how much preparation is enough, we need to take a closer look at the plans and policies about the emergency distribution of vaccine and medicine. 

Due to the undesirable side-effects of the vaccine, the routine vaccination for the variola virus has been discontinued ever since 1972 and currently the vaccination is only given to selected military personnel and laboratory workers who handle the virus. Moreover, because of the waning immunity of the vaccine, the proportion of Americans who are both over age 40 and still immune to smallpox might be too small to achieve herd immunity. As a result based on the stated concerns, public health authorities, such as Centers for Disease Control and Prevention (CDC), suggest intensive surveillance and identification of infected cases, isolation of smallpox patients, and vaccination of close contacts of infected individuals. In subsections $5.1$-$5.4$, we compare the effectiveness of two major post-event vaccination strategies: ring vaccination and mass vaccination, in response of a hypothetical bioterrorist smallpox attack in a big city. In subsection $5.5$, we assess how efficacies of the control strategies affect the result of ring vaccination.

\subsection{Ring Vaccination}
Also known as surveillance and containment, ring vaccination consists of rapid identification, isolation, vaccination of close contacts of infected persons (primary contacts), and vaccination of contacts of the primary contact (secondary contacts). We assume that each identified individual will be asked to provide a list of contacts of an average number (denoted by $CT$ as in the following text). Contacts that are successfully traced will be vaccinated and put under surveillance for a certain quarantine period. That is, vaccination and surveillance are follow-up procedures of tracing, and are applied to both susceptible contacts (who are in the quarantine class) and infected contacts (who are in the contact tracing class). We divide the population at time $t$ into seven classes as shown in a flow diagram in \emph{Fig.} \ref{SEIR}.
\begin{figure}[h]
\begin{center}
\includegraphics[scale=0.9, trim=4.0cm 19cm 0cm 4.5cm]{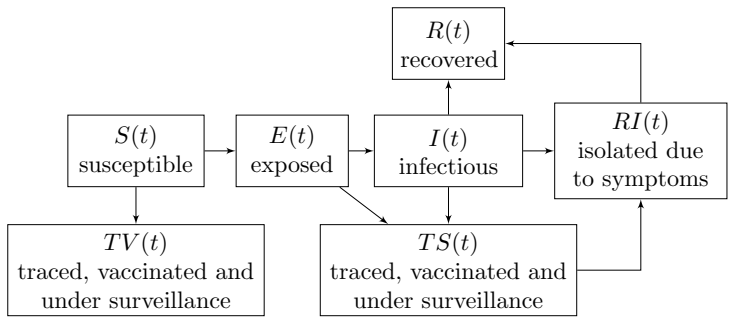}
\end{center}
\caption{\label{SEIR}Susceptible individuals may become infected, immediately enter the exposed class, $E(t)$, in which they are not yet infectious, and then enter the infectious class, $I(t)$. Symptomatic infectives in $I(t)$ may exit to the isolated class, $RI(t)$, and they will recover, enter $R(t)$, and do not return to $S(t)$ because of immunity. Susceptible people may be traced, vaccinated, put under surveillance, and they do not return to $S(t)$ because of vaccination. Exposed and infected individuals may be traced, vaccinated, put under surveillance, and then isolated when they become symptomatic. There is a possibility for infectious individuals who are neither identified due to symptoms or isolated by contact tracing, to enter class $R(t)$ when they reach the maximum disease age.}
\end{figure}

In the following context, we denote ${T_i}$ as the length of the pre infectious period (infectiousness threshold), ${T_s}$ as the length of the pre symptomatic period (symptoms threshold), and ${F_i}$ as the length of the infectious period. Hence ${T_i} + {F_i}$ represents the maximum disease age. We consider problem \eqref{whole model} with the same notations. Notice that the dynamics of $4$ compartments illustrated in \emph{Fig.} \ref{SEIR} depend on $S(t)$ and ${i\left( {a,t} \right)}$, since
$E\left( t \right) = \int_0^{{T_i}} {i\left( {a,t} \right)da} $, $I\left( t \right) = \int_{{T_i}}^{{T_i} + {F_i}} {i\left( {a,t} \right)da} $. Moreover, we have
\begin{flalign*}
\frac{{dR\!\left( t \right)}}{{dt}} & + \frac{{dRI\!\left( t \right)}}{{dt}} + \frac{{dTS\!\left( t \right)}}{{dt}} + \frac{{dTV\!\left( t \right)}}{{dt}}=  \underbrace{i\left( {{T_i} + {F_i},t} \right)}_{\substack{\text{recovery rate of}\\\text {unidentified infectives}}}\\
& + \underbrace{\int_{T_s}^{{T_i} + {F_i}}\!\!\!\!\!\!\! {\mu \left( a \right)\!i\left( {a,t} \right)\!da}}_{\substack{\text{isolation rate of}\\\text{symptomatic cases}}} + \underbrace{\mathcal{T}\!\left( {i\left( { \cdot ,t} \right)} \right)\!\left( {E\!\left( t \right)\! + \!I\!\left( t \right)} \right)}_{\substack{\text{tracing and surveillance}\\\text {rate of infectives}}} + \underbrace{\mathcal{Q}\!\left( {i\left( { \cdot ,t} \right)} \right)\!S\!\left( t \right)}_{\substack{\text{tracing and vaccination}\\\text {rate of susceptibles}}}
\end{flalign*}

Since tracing is a consequence of identifying symptomatic cases, then the number of contacts traced should be related to the number of infectious cases identified. Then we assume, for simplicity, that the tracing (hence vaccinating and surveilling) rate is proportional to the isolation (identifying new cases) rate. That is, in model \eqref{whole model}, we set
\begin{flalign}\label{TQB}
\begin{split}
&\mathcal{T}\!\left( {i\left( { \cdot ,t} \right)} \right) = {\eta _I}\cdot \frac{{CT}}{{{S_0}}}\int_{{T_s}}^{{T_i} + {F_i}} \!{\mu \left( a \right)i\left( {a,t} \right)da}\\
&\mathcal{Q}\!\left( {i\left( { \cdot ,t} \right)} \right) = {\eta _S}\cdot \frac{{CT}}{{{S_0}}} \int_{{T_s}}^{{T_i} + {F_i}} \!{\mu \left( a \right)i\left( {a,t} \right)da}\\
&\mathcal{B}\!\left( {i\left( { \cdot ,t} \right)} \right) = \int_{{T_i}}^{{T_i} + {F_i}} \!{\beta \left( a \right)i\left( {a,t} \right)da}
\end{split}
\end{flalign}
where ${\mu \!\left( a \right)}$ is the isolation removal rate for symptomatic infectives at disease age $a$ as in \eqref{whole model}, ${\beta \!\left( a \right)}$ is the disease transmission rate of an infectious individual at disease age $a$, $CT$ is the average number of contacts provided by each identified infective, ${S_0}$ is the initial susceptible population as in \eqref{whole model}, and ${\eta _I}$ and ${\eta _S}$ are proportionality constants for tracing the infected class and the susceptible class, respectively. Discussions about meanings and estimations of the parameters ${\eta _I}$ and ${\eta _S}$ are in the following context.\\

Contacts provided by an identified infective may come from any of the seven classes in \emph{Fig.} \ref{SEIR}, but only those who are in the classes $S(t)$, $E(t)$, and ${I(t)}$ may be successfully traced, vaccinated, and put under surveillance. We assume that the probability for a contact being infected (or susceptible) at time $t$ is proportional to the density of the infected (susceptible) population at time $t$, and we take ${\eta _I}$(${\eta _S}$) to be the constants of proportionality, respectively. Then at time $t$, the rate of tracing infected individuals is:
\begin{equation}\label{etainfected}
\underbrace {\underbrace {{\eta _I}\frac{{E\left( t \right) + I\left( t \right)}}{{{S_0}}}}_{\substack{\text{probability of tracing an infected contact}} }\cdot \; CT}_{\substack{\text{average number of infected contacts}\\\text{traced per identified symptomatic case}}} \cdot \underbrace {\left( {\int_{{T_s}}^{{T_i} + {F_i}} {\mu \left( a \right)i\left( {a,t} \right)da} } \right)}_{\substack{\text{rate of identifying symptomatic cases}}}
\end{equation}
The rate of tracing susceptible individuals is:
\begin{equation}\label{etasusceptible}
\underbrace {\underbrace {{\eta _S}\frac{{S\left( t \right)}}{{{S_0}}}}_{\substack{\text{probability of tracing}\\\text{ a susceptible contact}}} \cdot \; CT}_{\substack{\text{average number of susceptible contacts}\\\text{traced per identified symptomatic case}}} \cdot \underbrace {\left( {\int_{{T_s}}^{{T_i} + {F_i}} {\mu \left( a \right)i\left( {a,t} \right)da} } \right)}_{\text{rate of identifying symptomatic cases}}
\end{equation}

which are exactly the corresponding terms in \eqref{whole model} with the setting \eqref{TQB}. Moreover, the probability interpretations in \eqref{etainfected} and \eqref{etasusceptible} imply that for any time $t$, ${\eta _I}$ and ${\eta _S}$ should satisfy:
\begin{equation}\label{eta}
{\eta _I}\frac{{E\left( t \right) + I\left( t \right)}}{{{S_0}}} + {\eta _S}\frac{{S\left( t \right)}}{{{S_0}}} \leqslant 1
\end{equation}

We denote the probability that a traced contact of an identified symptomatic individual is infected as $\eta$, a parameter that describes the tracing efficacy in finding infectives. Hence ${\eta_I}$ and ${\eta_S}$ can be obtained when the value of $\eta$ is given: ${\eta _I} = \eta \frac{{{S_0}}}{{E\left( t \right) + I\left( t \right)}}$ from \eqref{etainfected}, and ${\eta _S} \leqslant \left( {1 - \eta } \right)\frac{{{S_0}}}{{S\left( t \right)}}$ by \eqref{etasusceptible}. Since $S(t)$ is mostly unchanged for $t$ in the initial phase of the outbreak, for simplicity, we take $\frac{{{S_0}}}{{S\left( t \right)}} \approx 1$ and $\frac{{{S_0}}}{{E\left( t \right) + I\left( t \right)}} \approx \frac{{{S_0}}}{{E\left( 0 \right) + I\left( 0 \right)}}$ for any time $t$. So ${\eta_I}$ and ${\eta_S}$ are estimated by $\eta$ and the initial conditions, \emph{i.e.}, ${\eta _I} = \eta \frac{{{S_0}}}{{E\left( 0 \right) + I\left( 0 \right)}}$ and ${\eta _S} \leqslant 1 - \eta $. The value of $\eta$ can be easily determined from evolving data during the initial phase of the epidemic: it is simply the fraction of the traced contacts who turn out to be symptomatic over all traced contacts.

In particular, when ${\eta _I} = {\eta _S} = 1$, then the probability of tracing an infected contact at time $t$ is $\eta  = \frac{{E\left( t \right) + I\left( t \right)}}{{{S_0}}}$ and that of tracing a susceptible contact at time $t$ is $\frac{{S\left( t \right)}}{{{S_0}}}$. That means the probability of tracing an infected (susceptible) contact at time $t$ is exactly the fraction of infected (susceptible) population at time $t$, which indicates that the tracing is random and is not effective.

\subsection{Mass Vaccination}
Mass vaccination, usually conducted at a constant rate, is the strategy of vaccinating large numbers of people. We assume that there is no residual immunity in the population, and a post-event mass vaccination, together with a strategy of isolating symptomatic individuals, start as soon as the first case is identified. We consider the fact that infected people vaccinated in the first few days of exposure will not transmit smallpox to others (\citealt{CDCsymptom}). And we denote ${T_v}$ as the length of vaccine sensitive period for infectives, that is, infectives receive vaccination with disease age less than ${T_v}$ will not be infectious. This assumption is not relevant in ring vaccination: infected contacts at any disease age are removed due to vaccination and surveillance, hence ${T_v}$ is a parameter that only used in mass vaccination. \emph{Fig.} \ref{SEIR2} shows the dynamic of the disease transmission with mass vaccination.
\begin{figure}[here]
\begin{center}
\includegraphics[scale=0.9, trim=4.0cm 19cm 0cm 4cm]{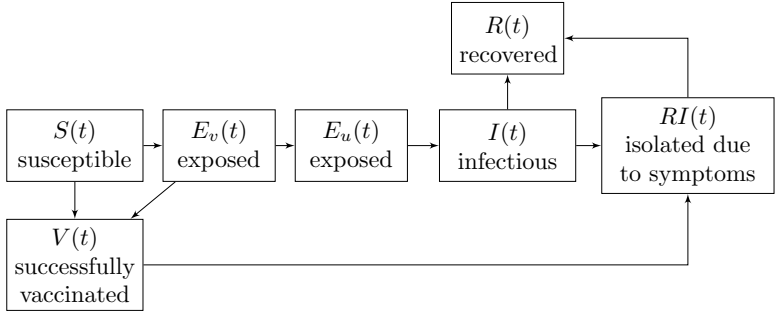}
\end{center}
\caption{\label{SEIR2}${E_v}(t)$ denotes the number of infectives in the vaccine sensitive stage, so they will not transmit the disease to others if vaccinated. ${E_u}(t)$ denotes the number of infectives in the vaccine insensitive stage, who will be able to transmit the disease even after vaccination.}
\end{figure}

The following mass vaccination model is of a simpler form than \eqref{whole model}, which can be analyzed by the method in (\citealt{webb2010}). 
\begin{flalign}\label{mass vaccination}
\begin{split}
&\frac{\partial }{{\partial t}}i\left( {a,t} \right) + \frac{\partial }{{\partial a}}i\left( {a,t} \right) =  - \mu \left( a \right)i\left( {a,t} \right) - \nu \left( a \right)i\left( {a,t} \right),\;0 \leqslant a \leqslant {T_i} + {F_i},\;t \geqslant 0\\
&\frac{d}{{dt}}S\left( t \right) =  - \left( {\int_{{T_i}}^{{T_i} + {F_i}} {\beta \left( a \right)i\left( {a,t} \right)da} } \right)S\left( t \right) - {\nu_0} S\left( t \right),\;t \geqslant 0\\
&i\left( {0,t} \right) = \left( {\int_{{T_i}}^{{T_i} + {F_i}} {\beta \left( a \right)i\left( {a,t} \right)da} } \right)S\left( t \right),\;t \geqslant 0\\
&i\left( {a,0} \right) = {i_0}\left( a \right),\;0 \leqslant a \leqslant {T_i} + {F_i},\;S\left( 0 \right) = {S_0}
\end{split}
\end{flalign}
where $\nu \left( a \right)$ is the mass vaccination removal rate of infected individuals at disease age $a$, ${\nu_0}$ is the mass vaccination rate, and the other notations have the same interpretations as in the ring vaccination model.

\subsection{Model Parameters}
\emph{Table 1} describes smallpox natural history and \emph{Table 2} shows parameter values/ranges we choose for simulations. We pick the threshold values of ${T_i}$, ${T_s}$ and ${F_i}$ as recommended in (\citealt{Eichner}) and (\citealt{CDCsymptom}). \emph{Fig.} \ref{beta} illustrates the transmission rate function of disease age, the shape of the function suggested in studies (\citealt{integralmodel, carrat, Eichner, CDCsymptom}), and (\citealt{CCC}), and we make a theoretical estimation about the value of the transmission rate. We vary ${R_{sym}}$, the percentage of symptomatic individuals removed per day, from $\mathrm{50\%} $ to $\mathrm{90\%}$, which is an estimation of an efficient removal process of smallpox due to its identifiable symptoms after the prodrome.\\

We model a deliberate release of smallpox pathogen in a big city as large as New York, which has a total population of $\mathrm{8 \times {10^6}}$. All of the simulations start with an age distribution of index cases as shown in \emph{Fig.} \ref{initial}, which corresponds to a scenario when one or several public places encounter a series of smallpox virus releases.

%%%%%%%%%%%%%%%%%%%%%%%%%%%%%%%%%%%%%%%%%%%%%%%%%%%%%%%%%%%%%%%%%%%%%%%%%%%%%%%%%%%%%%%%%%%%%%%%%%%%%%%%%%%%%%%%%%%%%%%%%%
\begin{table}[!ht]
\begin{center}
\begin{tabular}{|c |c |c |c|}
\multicolumn{4}{l}{\emph{Table 1}. Smallpox durations of the progression stages.}\\
\hline
Stage & Duration &  Infectiousness & References\\
\hline
Incubation period & $7 \sim 17$ days &  Not infectious & (\citealt{CDCsymptom}) \\
\hline
Initial symptoms(prodrome) & $2 \sim 4$ days & Sometimes infectious & (\citealt{CDCsymptom}) \\
\hline
Early rash & 4 days & Most infectious & (\citealt{CDCsymptom}) \\
\hline
Pustular rash and scabs & 16 days & Infectious & (\citealt{CDCsymptom}) \\
\hline
Scabs resolved &       & Not infectious & (\citealt{CDCsymptom}) \\
\hline
\end{tabular}
\label{natural history}
\end{center}
\end{table}
%%%%%%%%%%%%%%%%%%%%%%%%%%%%%%%%%%%%%%%%%%%%%%%%%%%%%%%%%%%%%%%%%%%%%%%%%%%%%%%%%%%%%%%%%%%%%%%%%%%%%%%%%%%%%%%%%%%%%%%%%%%
\begin{table}
\begin{center}
\begin{tabular}{|p{2.6 cm} |p{5.2cm} |p{2.6cm}|}
\multicolumn{3}{l}{\emph{Table 2.} Baseline parameters and initial conditions.}\\
\hline
Parameter description & Parameter baseline value & References\\
\hline
infectiousness threshold & ${T_i} = 12$ days & (\citealt{Eichner})\\
\hline
symptoms threshold & ${T_s} = 14$ days & (\citealt{CDCsymptom, Eichner})\\
\hline
vaccine insensitiveness threshold$^\dag$ & ${T_v} = 3$ days & (\citealt{CDCsymptom})\\
\hline
length of infectious period & ${F_i} = 28$ days & (\citealt{CDCsymptom})\\
\hline
infection transmission rate function & $\beta \left( a \right) = \left\{ \begin{matrix}
  0\,,\quad \quad \quad \quad \quad \quad 0 \leqslant a < 12 \hfill \cr 
  2.5 \cdot {10^{ - 8}}{\left( {a - 12} \right)^2}{e^{ - 0.5\left( {a - 12} \right)}},\; \hfill \cr 
  \quad \quad \quad \quad \quad \quad \;12 \leqslant a \leqslant 40 \hfill \cr 
 \end{matrix}  \right.$ & \emph{Fig.} \ref{beta}\\
\hline
removal of symptomatic cases & ${R_{sym}} \geqslant 50\% $ per day & (\citealt{webb2010, Meltzer})\\
\hline
isolation rate of infectives & $\mu \left( a \right) = \left\{ \begin{matrix} 
  0.0,\quad \quad \quad \quad \quad a \leqslant {14} \hfill \cr 
   - \ln \left( {1.0 - {R_{sym}}} \right),a > {14} \hfill \cr 
 \end{matrix}  \right.$ & (\citealt{webb2010})\\
\hline
mass vaccination removal rate of infectives$^\dag$ & $\nu \left( a \right) = \left\{ \begin{matrix} 
  {\nu_0} ,\;0 \leqslant a \leqslant 3 \hfill \cr 
  0,\;a > 3 \hfill \cr 
 \end{matrix}  \right.$ & (\citealt{CDCsymptom})\\
\hline
mass vaccination rate$^\dag$ & $0 \leqslant {\nu_0} \leqslant 10000$ & (\citealt{Kaplan})\\
\hline
average number of contacts traced per identified case$^*$ & $0 < CT \leqslant 100$ & (\citealt{Kaplan})\\
\hline
probability for a traced contact being infected$^*$ & $0 < \eta  \leqslant 1$ & Text\\
\hline
initial susceptible population & ${S_0} = 8 \times {10^6}$ & Text\\
\hline
index cases distribution & ${i_0} \in L_ + ^1\left[ {0,{T_i} + {F_i}} \right]$ & \emph{Fig.} \ref{initial}\\
\hline
\end{tabular}
\caption*{$^*$parameters only used in ring vaccination.\\ 
$^\dag$parameters only used in mass vaccination.}
\end{center}

\end{table}
%%%%%%%%%%%%%%%%%%%%%%%%%%%%%%%%%%%%%%%%%%%%%%%%%%%%%%%%%%%%%%%%%%%%%%%%%%%%%%%%%%%%%%%%%%%%%%%%%%%%%%%%%%%%%%%%%%%%%%%%%%%%%
 \begin{figure}[!ht]
    \subfloat[Infection transmission rate function $\beta$ \label{beta}]{%
      \includegraphics[width=0.45\textwidth]{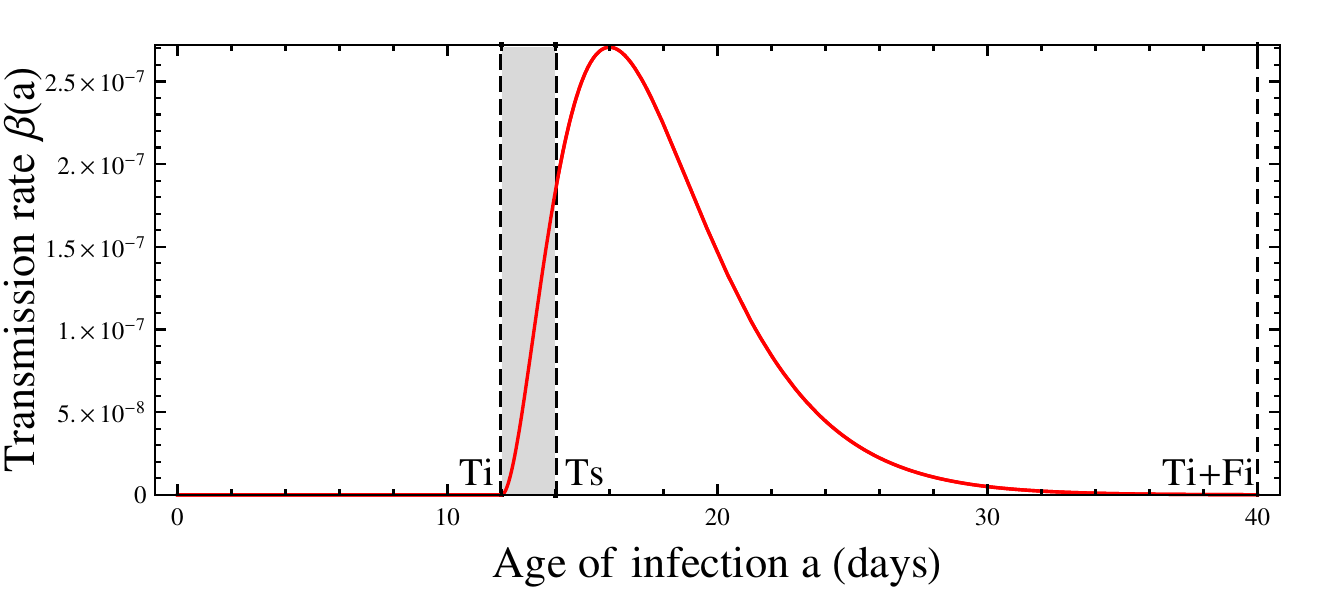}
    }
    \hfill
    \subfloat[Initial disease-age distribution function ${i_0}$ \label{initial}]{%
      \includegraphics[width=0.45\textwidth]{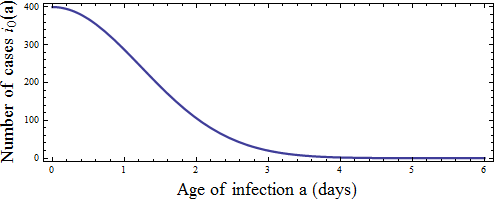}
    }
    \caption{
    \label{smallpoxinitialcondition} \emph{Fig.} \ref{beta} is the infection age-dependent transmission rate function with respect to age since infection. ${T_i}$, ${T_s}$, and ${F_i}$ are introduced and estimated as in the context and \emph{Table 2}. The grey area is the prodromal period with initial symptoms and early contagiousness. \emph{Fig.} \ref{initial} is the initial disease distribution function ${i_0}$.}
  \end{figure}
%%%%%%%%%%%%%%%%%%%%%%%%%%%%%%%%%%%%%%%%%%%%%%%%%%%%%%%%%%%%%%%%%%%%%%%%%%%%%%%%%%%%%%%%%%%%%%%%%%%%%%%%%%%%%%%%%%%%%%%%%%%%%

%%%%%%%%%%%%%%%%%%%%%%%%%%%%%%%%%%%%%%%%%%%%%%%%%%%%%%%%%%%%%%%%%%%%%%%%%%%%%%%%%%%%%%%%%%%%%%%%%%%%%%%%%%%%%%%%%%%%%%%%%%%%%

%%%%%%%%%%%%%%%%%%%%%%%%%%%%%%%%%%%%%%%%%%%%%%%%%%%%%%%%%%%%%%%%%%%%%%%%%%%%%%%%%%%%%%%%%%%%%%%%%%%%%%%%%%%%%%%%%%%%%%%%%%%%%%%%%%%%%%%%%%%%%%%%%%%%%%%%%%%%%%%%%%%%%%%%%%%%%%%%%%%%%%%%%%%%%%%%%%%%%%%%%%%%%%%%%%%%%%%%%%%%%%%%%%%%%%%%%%%%%%%%%%%%%%%%%%%%%%%

\subsection{Simulations of Different Vaccination Strategies: Mass Vaccination versus Ring Vaccination}
There are two ring vaccination scenarios in \emph{Fig.} \ref{fourcases}: the green curves represent an effective ring vaccination strategy, and the blue curves represent an ineffective ring vaccination strategy when ${\eta_I} = {\eta_S} = 1$. We observe pulses in the daily number of traced and vaccinated contacts in both of the two scenarios. These pulses are caused by the choice of the initial infection-age distribution function in \emph{Fig.} \ref{initial}. The majority of the index cases are in an early disease age, and thus they will become infectious and symptomatic in the same time period. As a consequence, symptomatic cases and generations of new cases will appear as pulses; hence daily traced contacts will appear as pulses, since the tracing rate depends on the isolation rate of symptomatic cases.\\

\begin{figure}
\begin{center}
\includegraphics[scale=0.36]{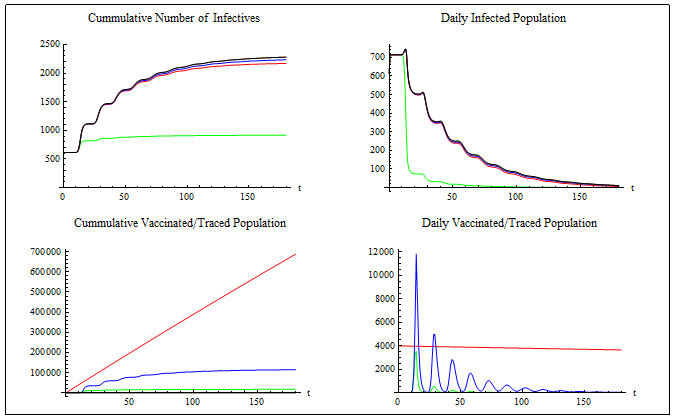}
\caption{\label{fourcases}There are four different cases included in this figure. The green curves stand for the case with effective ring vaccination when $CT = 50$, ${R_{sym}}=50\%$, and $\eta = 0.1$. The blue curves represent the case with ineffective ring vaccination when $CT =50$, ${R_{sym}}=50\%$, and ${\eta_I} = {\eta_S} = 1$, so $\eta  = \frac{{I\left( t \right) + E\left( t \right)}}{{{S_0}}}$. The red curves are for the case with mass vaccination when ${R_{sym}}=50\%$ and the mass vaccination rate is $4000$ individuals per day. The black curves correspond to the case with an isolation removal rate of ${R_{sym}}=50\%$ but no vaccination.}
\end{center}
\end{figure}

\emph{Fig.} \ref{fourcases} also gives comparison between ring and mass vaccination strategies from different aspects: (1) the effective ring vaccination strategy prevents the most cases from occurring and requires less personnel and less vaccine stockpiles; (2) the effective ring vaccination strategy does not require a large number of people to be traced everyday, and is more efficient in controlling the outbreak compared to the mass vaccination (red curves), which requires vaccination of a large number of people everyday; (3) the ineffective ring vaccination has similar results in controlling the outbreak as the mass vaccination (red curves), even though it consumes less vaccine stockpiles in total; it requires extremely heavy daily contact tracing load at times; (4) compared with the case of no vaccination (black curves), mass vaccination and ineffective ring vaccination prevent hundreds of cases from happening; (5) further simulations show that, for higher ${R_{sym}}$ values, non-vaccination could control the spread of smallpox as well as mass vaccination and ineffective ring vaccination, while in contrast effective ring vaccination attains significant improvement in reducing total number of cases.\\

We also take into consideration the fact that tracing, vaccinating, and surveilling a contact in ring vaccination demands a different level of personnel effort than in mass vaccination. So in \emph{Fig.} \ref{twocases}, we compare an effective ring vaccination of a highest daily contact tracing rate of $\mathrm{4,000}$ contacts per day, with a mass vaccination of a constant daily vaccination rate of $\mathrm{12,000}$ people per day. That is, we assume that tracing, vaccinating, and surveilling a contact requires three times more effort than the comparable mass vaccination effort. As can be seen from the simulation, the effective ring vaccination prevents more cases, and vaccinates less people than the mass vaccination, which would also help reduce serious vaccination side effects.

\begin{figure}
\begin{center}
\includegraphics[scale=0.36]{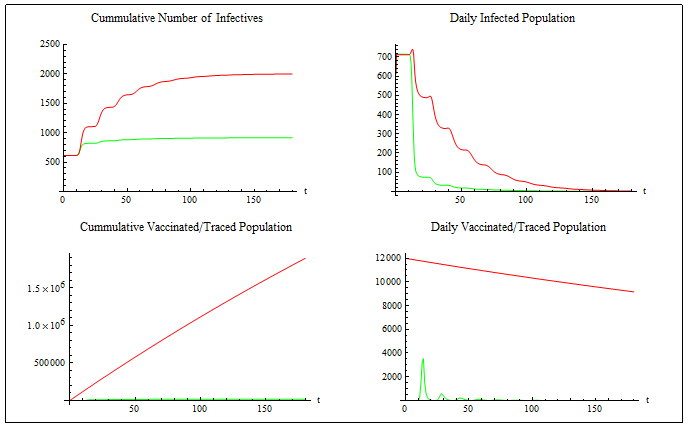}
\caption{\label{twocases} The green curves stand for the case of an effective ring vaccination when $CT = 50$, ${R_{sym}}=50\%$, and $\eta = 0.1$. The red curves are for the case of mass vaccination when ${R_{sym}}=50\%$ and the mass vaccination rate is $\mathrm{1,2000}$ individuals per day.}
\end{center}
\end{figure}

%%%%%%%%%%%%%%%%%%%%%%%%%%%%%%%%%%%%%%%%%%%%%%%%%%%%%%%%%%%%%%%%%%%%%%%%%%%%%%%%%%%%%%%%%%%%%%%%%%%%%%%%%%%%%%%%%%%%%%%%%%%%%%%%%%%%%%%%%%%%%%%%%%%%%%%%%%%%%%%%%%%%%%%%%%%%%%%%%%%%%%%%%%%%%%%%%%%%%%%%%%%%%%%%%%%%%%%%%%%%%%%%%%%%%%%%%%%%%%%%%%%%%%%%%%%%%%%%%

\subsection{Simulations of Ring Vaccination: Assessing Impacts of Parameters}
In order to provide guidance to public health authorities for containment and surveillance strategies, we vary the three variables, $CT$, ${\eta}$, and ${R_{sym}}$, to assess different levels of ring vaccination by evaluating: (1) total number of infected cases, and (2) the percentage of traced individuals.\\
\begin{figure}
\begin{center}
\includegraphics[scale=0.37, trim= 0mm 30mm 0mm 0mm]{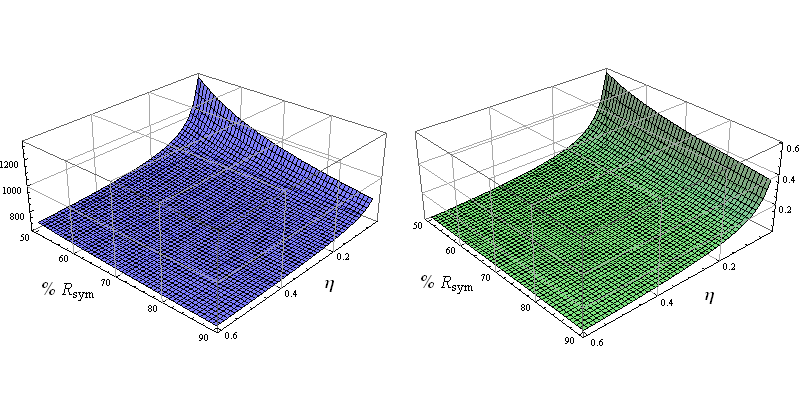}
\caption*{\footnotesize{total number of cases\quad \quad \quad\quad \quad \quad \quad \quad \quad \quad $\%$ of people traced}}
\caption{\label{Rsymeta} The blue surface is the total number of cases as a function of daily removal percentage of symptomatic cases $50\%  \leqslant {R_{sym}} \leqslant 90\% $, and the probability that a traced contact of an identified symptomatic individual is infected $0 < {\eta} \leqslant 0.6$. The number of contacts traced per identified case is $CT = 50$. The green surface is the percentage of contact traced individuals as a function of the same variables under the same settings.}
\end{center}
\end{figure}

The simulation results in \emph{Fig.} \ref{Rsymeta} are intuitively reasonable: for fixed $CT = 50$, high efficacies of both isolation and contact tracing will prevent more cases and save more personnel engaged in tracing. Increasing $\eta$ enables us to trace more infected contacts per identified case, and it in turn saves personnel efforts. When the values of ${\eta}$ and ${R_{sym}}$ are relatively small, increasing either one of them is efficient in both controlling the outbreak and relieving the burden of tracing. If we are already able to maintain the isolation and contact tracing at a relatively high level, increasing either of the two levels would require more personnel to be involved, but just improve the results slightly.

\begin{figure}
\begin{center}
\includegraphics[scale=0.36, trim = 0mm 30mm 0mm 0mm]{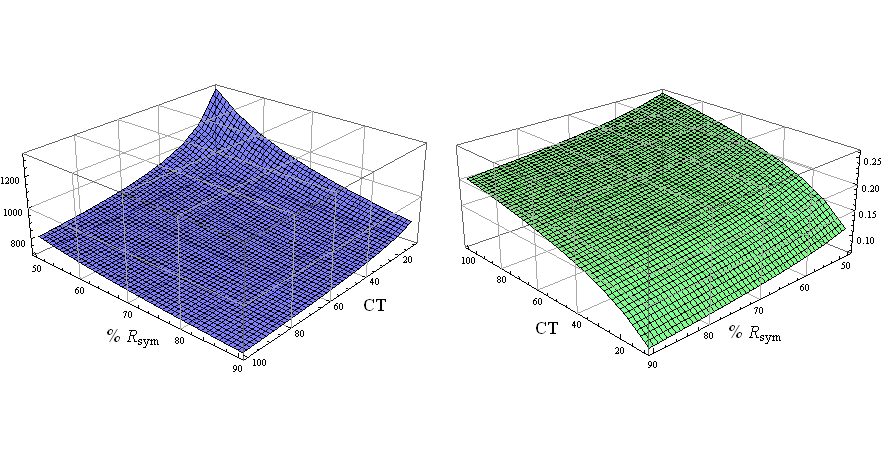}
\caption*{\footnotesize{total number of cases\quad \quad \quad\quad \quad \quad \quad \quad \quad \quad $\%$ of people traced}}
\caption{\label{RsymCT} The blue surface is total number of infected cases as a function of daily removal percentage of symptomatic cases $50\%  \leqslant {R_{sym}} \leqslant 90\% $, and the average number of contacts traced per identified case $10 \leqslant CT \leqslant 100$. The probability for tracing an infected contact is fixed as we take $\eta = 0.1$. The green surface is the percentage of contact traced population as a function of the same variables under the same settings.}
\end{center}
\end{figure}

We fix ${\eta} = 0.1$ in \emph{Fig.} \ref{RsymCT}, and notice that raising the value of $CT$ does help reduce the total number of cases, but it also boosts the demand for the number of health care workers engaged in tracing, vaccinating, and surveilling. For fixed value of $\eta$, increasing $CT$ helps reduce infections in two ways: increasing the number of infected contacts traced per identified case; and increasing the number of susceptible contacts quarantined which results in lower infection rates. Since large values of $CT$ and ${R_{sym}}$ require more public health resources, it is left to the public health officials to determine appropriate levels of contact tracing and isolation. In the case when an effective vaccination is absent, we do not expect to quarantine a great amount of susceptibles, so the decision of increasing $CT$ should be carefully made.\\

\begin{figure}
\begin{center}
\includegraphics[scale=0.35, trim = 0mm 25mm 0mm 0mm]{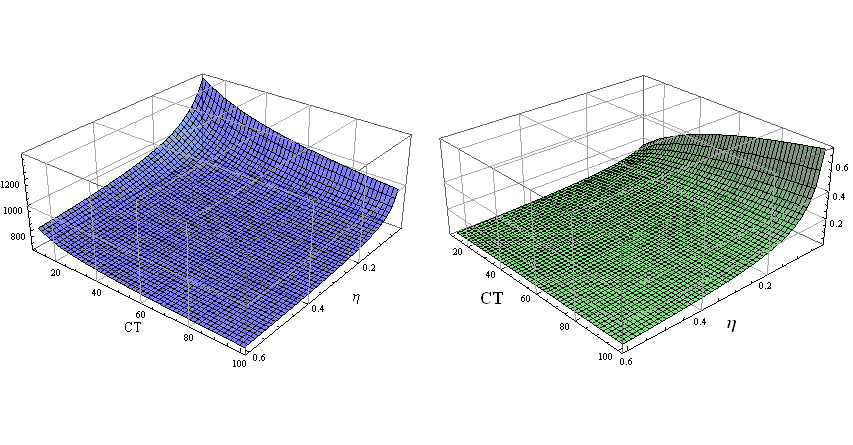}
\caption*{\footnotesize{total number of cases\quad \quad \quad\quad \quad \quad \quad \quad \quad \quad $\%$ of people traced}}
\caption{\label{CTeta} The blue surface is attack ratio as a function of the probability that a traced contact of an identified symptomatic individual is infected $0 < {\eta} \leqslant 0.6$, and the number of contacts traced per identified case $10 \leqslant CT \leqslant 100$. The daily removal percentage of symptomatic cases is fixed as ${R_{sym}} = 60\% $. The green surface is the percentage of contact traced population as a function of the same variables under the same settings.}
\end{center}
\end{figure}

In \emph{Fig.} \ref{CTeta}, we assume that the removal percentage of symptomatic cases is fixed as $\mathrm{60\%}$. $CT$ and $\eta$ represent different aspects of ring vaccination strategy, and this simulation suggests how to deploy resources assigned in tracing and control the outbreak in a more economical way. In contrast to \emph{Fig.} \ref{RsymCT}, when contact tracing is of higher efficacy in finding infected contacts, increasing the average number of contacts provided by each identified symptomatic case does not boost greatly the demand for personnel and vaccine stockpiles. So under the assumption that the tracing efficiency $\eta$ can be maintained while $CT$ is increased, tracing more contacts per case will help prevent cases and will not result in much more tracing work.
%%%%%%%%%%%%%%%%%%%%%%%%%%%%%%%%%%%%%%%%%%%%%%%%%%%%%%%%%%%%%%%%%%%%%%%%%%%%%%%%%%%%%%%%%%%%%%%%%%%%%%%%%%%%%%%%%%%%%%%%%%%%%%%%%%%%%%%%%%%%%%%%%%%%%%%%%%%%%%%%%%%%%%%%%%%%%%%%%%%%%%%%%%%%%%%%%%%%%%%%%%%%%%%%%%%%%%%%%%%%%%%%%%%%%%%%%%%%%%%%%%%%%%%%%%%%%%%%%%%%%%%%%%%%%%%%%%%%%%%%%%%%%%%%%%%%%%%%%%%%%%%%%%%%%%%%%%%%%%%%%%%%%%%%%%%%%%%%%%%%%%%%%%%%%%%%%%%%%%%%%%%%%%%%%%%%%%%%%%%%%%%%%
\section{Application II: Influenzas: SARS}
\label{sec:6}
In this section, we apply our model to investigate contact tracing effectiveness in control of modern influenzas and take the outbreak of SARS (severe acute respiratory syndrome) as an example. First, we use our model to simulate the SARS outbreak in Taiwan, 2003 by real data. Then we modify the length of presymptomatic period hypothetically in order to provide suggestions about the efficacy of contact tracing under different circumstances.\\

Distinct from smallpox, the conduct of surveillance and control strategies of modern influenzas (such as H1N1 and SARS) is less efficient due to the lack of timely vaccines, non-compliance of the public with quarantine, and period of asymptomatic infectiousness. It is widely believed that SARS was eradicated because of limited transmission occurring before symptom onset, but the effectiveness of contact tracing is still controversial even in the regions where high levels of contact tracing were conducted, such as Taiwan and mainland China. We apply our model to simulate the SARS outbreak in Taiwan, 2003, concentrated in the Taipei-Keelung metropolitan area (with a population of 6 million in 2003), because of the extensive available data and the high efficiency of contact tracing in Taiwan.\\

\subsection{Parameters}
In \emph{Table 3}, we list the parameters that are obtained from real data (\citealt{MMWR}) and clinical studies of SARS. In (\citealt{MMWR}), we count the total number of traced close contacts as $\mathrm{12,394}$. Of those there are only $\mathrm{33}$ confirmed to be infected. In this way we set $\eta  \approx {{\mathrm{33}} \mathord{\left/
 {\vphantom {{\mathrm{33}} {\mathrm{12394}}}} \right.
 \kern-\nulldelimiterspace} {\mathrm{12394}}}$ and $CT \approx {{\mathrm{12394}} \mathord{\left/
 {\vphantom {{\mathrm{12394}} {\mathrm{300}}}} \right.
 \kern-\nulldelimiterspace} {\mathrm{300}}}$, where $\mathrm{300}$ is approximately the total number of cases in Taipei-Keelung metropolitan area (which we refer to Taipei area for short in the following context). The setting of the baseline values has three uncertain aspects: (1) ${T_s}-{T_i}$ determines the length of the incubation period: we set that to be $\mathrm{1.0}$ day in the data fitting in subsection $6.2$. (2) We make an assumption of the shape of transmission rate function $\beta$ based on laboratory diagnosis of SARS (\citealt{carrat, Clinical, Laboratory}), and determine its value by estimating the basic reproduction number of SARS in Taiwan as about $\mathrm{4.8}$ (\citealt{Hsieh3, Galvani}). (3) We estimate the shape of the removal rate function of symptomatic individuals $\mu$ by comparing it to studies (\citealt{Brauer, Nishiura, carrat}). We estimate the maximal value of $\mu$ in \emph{Fig.} \ref{SARSinitialmu} by fitting data of SARS in the Taipei area.
%%%%%%%%%%%%%%%%%%%%%%%%%%%%%%%%%%%%%%%%%%%%%%%%%%%%%%%%%%%%%%%%%%%%%%%%%%%%%%%%%%%%%%%%%%%%%%%%%%%%%%%%%%%%%%%%%%%%%%%%%%%
\begin{table}[here]
\begin{center}
\begin{tabular}{|p{4.0cm} |p{2.2cm} |p{4.2cm}|}
\multicolumn{3}{l}{\emph{Table 3.} Baseline Parameters}\\
\hline
Parameter description & Baseline values & References\\
\hline
infectiousness threshold & ${T_i} = 5.0$ days & (\citealt{incubation1, incubation2})\\
\hline
symptoms threshold & ${T_s} = 6.0$ days & (\citealt{incubation1, incubation2})\\
\hline
length of infectious period & ${F_i} = 21.0$ days & (\citealt{incubation1, incubation2})\\
\hline
isolation rate of infectives & \emph{Fig.} \ref{SARSinitialmu} & (\citealt{Brauer, Nishiura})\\
\hline
infection transmission rate function & \emph{Fig.} \ref{SARSinitialbeta} & (\citealt{Hsieh3, Galvani, carrat, Clinical, Laboratory})\\
\hline
average number of contacts traced per identified case & $CT \approx {{\mathrm{12394}} \mathord{\left/
 {\vphantom {{\mathrm{12394}} {\mathrm{300}}}} \right.
 \kern-\nulldelimiterspace} {\mathrm{300}}}$ & (\citealt{MMWR})\\
\hline
probability for a traced contact being infected & $\eta  \approx {{\mathrm{33}} \mathord{\left/
 {\vphantom {{\mathrm{33}} {\mathrm{12394}}}} \right.
 \kern-\nulldelimiterspace} {\mathrm{12394}}}$ & (\citealt{MMWR})\\
\hline
initial susceptible population & ${S_0} = 6 \times {10^6}$ & Text\\
\hline
index cases distribution & \emph{Fig.} \ref{SARSinitiali0} & (\citealt{Hsieh3, MMWR}) \\
\hline
\end{tabular}
\end{center}

\end{table}
%%%%%%%%%%%%%%%%%%%%%%%%%%%%%%%%%%%%%%%%%%%%%%%%%%%%%%%%%%%%%%%%%%%%%%%%%%%%%%%%%%%%%%%%%%%%%%%%%%%%%%%%%%%%%%%%%%%%%%%%%%%%%
 \begin{figure}[!ht]
    \subfloat[Infection transmission rate function $\beta$ \label{SARSinitialbeta}]{%
      \includegraphics[width=0.315\textwidth]{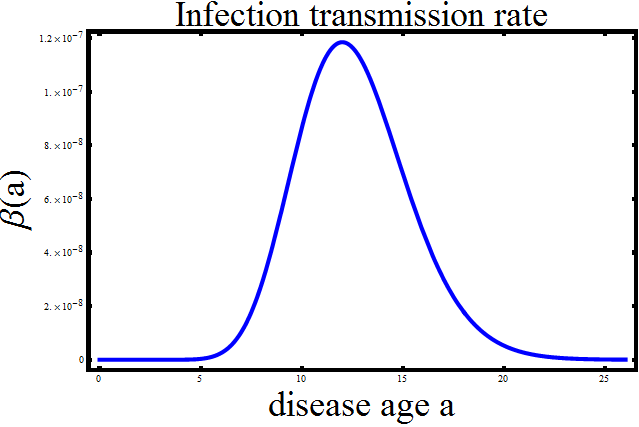}
    }
    \hfill
    \subfloat[Removal rate function $\mu$ \label{SARSinitialmu}]{%
      \includegraphics[width=0.3\textwidth]{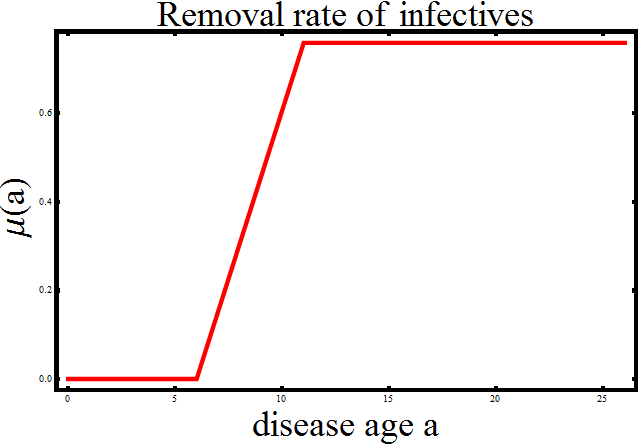}
    }
    \hfill
    \subfloat[Initial disease age distribution ${i_0}$ \label{SARSinitiali0}]{%
      \includegraphics[width=0.3\textwidth]{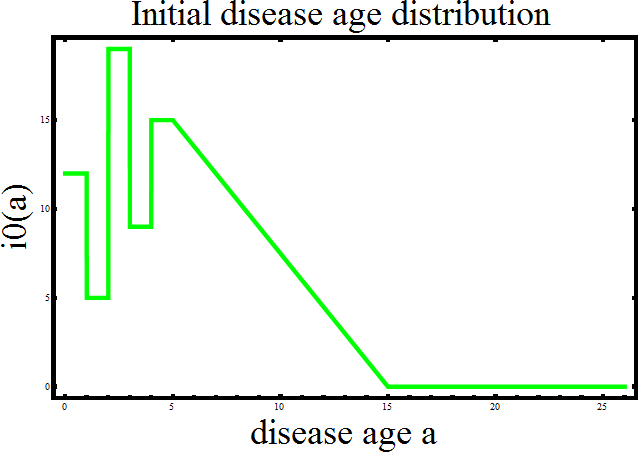}
    }
    \caption{
    \label{SARSinitialcondition} The $\beta$, $\mu$, and ${i_0}$ functions we use to fit the SARS data in Taiwan, 2003.}
  \end{figure}
%%%%%%%%%%%%%%%%%%%%%%%%%%%%%%%%%%%%%%%%%%%%%%%%%%%%%%%%%%%%%%%%%%%%%%%%%%%%%%%%%%%%%%%%%%%%%%%%%%%%%%%%%%%%%%%%%%%%%%%%%%%%%

\subsection{Data Fitting}
Our simulation results are shown in \emph{Fig.} \ref{datafit}. Compared to simply using isolation of symptomatic cases, enforced contact tracing can help prevent $\mathrm{40}$ individuals from being infected. That means contact tracing and quarantine of more than $\mathrm{12,000}$ people in the Taipei area enabled public health officials to discover $\mathrm{33}$ infected cases, and thus about $\mathrm{7}$ cases were avoided.\\
\begin{figure}[h]
\begin{center}
\includegraphics[scale=0.36]{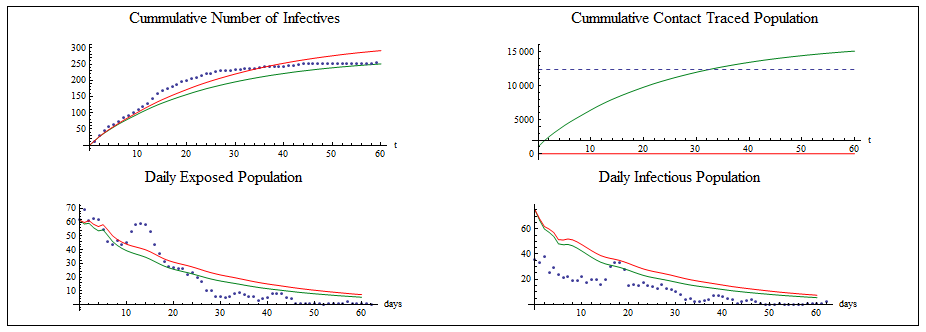}
\caption{\label{datafit} The blue dots are the real data from the SARS outbreak in Taipei-Keelung metropolitan area, 2003. The green curves represent our simulation to fit the real data with parameters in \emph{Table 3}. The red curves stand for the assumption that there is only isolation of symptomatic cases but no contact tracing implemented.}
\end{center}
\end{figure}

\subsection{Alternative Fitting Parameters}
In \emph{Fig.} \ref{presymptomatic}, we modify two of the uncertain parameters mentioned before, the length of presymptomatic period ${T_s}-{T_i}$ and the removal rate of symptomatic infectives (which is the highest value of the removal rate function $\mu$), and fit the data from Taiwan SARS. As can be inferred from \emph{Fig.} \ref{presymptomatic}, longer incubation period 
requires higher efficiency of symptomatic case isolation in order to maintain the total number of cases at the same level. As a consequent result, we observe that the number of cases that are avoided by contact tracing under all pairs of parameters in \emph{Fig.} \ref{presymptomatic} is as much as $\mathrm{40}$. Which means, we can assess the effectiveness of contact tracing in SARS, Taiwan without an accurate estimation of the incubation period since different assumptions lead to similar results.
\begin{SCfigure}
\begin{centering}
\includegraphics[scale=0.5]{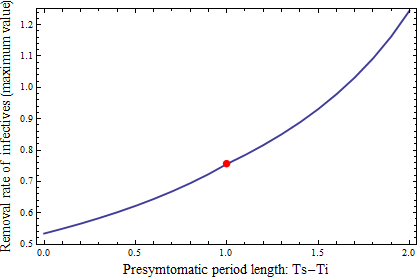}
\caption{\label{presymptomatic} The blue curve stands for pairs of parameters, incubation period length ${T_s}-{T_i}$ and maximum removal rate of symptomatic infectives, that can be applied in our model in order to fit the outbreak data of SARS, Taipei area. The red dot is the pair of parameter we used in the data fitting \emph{Fig.} \ref{datafit}. Further simulation indicates that the number of cases that are avoided by contact tracing remains the same no matter which pair of the parameters we take.}
\end{centering}
\end{SCfigure}
%%%%%%%%%%%%%%%%%%%%%%%%%%%%%%%%%%%%%%%%%%%%%%%%%%%%%%%%%%%%%%%%%%%%%%%%%%%%%%%%%%%%%%%%%%%%%%%%%%%%%%%%%%%%%%%%%%%%%%%%%%%%%%%%%%%%%%%%%%%%%%%%%%%%%%%%%%%%%%%%%%%%%%%%%%%%%%%%%%%%%%%%%%%%%%%%%%%%%%%%%%%%%%%%%%%%%%%%%%%%%%%%%%%%%%%%%%%%%%%%%%%%%%%%%%%%%%%%%%%%%%%%%%%%%%%%%%%%%%%%%%%%%%%%%%%%%%%%%%%%%%%%%%%%%%%%%%%%%%%%%%%%%%%%%%%%%%%%%%%%%%%%%%%%%%%%%%%%%%%%%%%%%%%%%%%%%%%%%%%%%%%%%
\section{Discussion}
\label{sec:7}
We introduce a general epidemic model that takes age since infection into consideration, to model interventions such as contact tracing, quarantine, and vaccination. Our model is applicable to different control strategies that can be formulated consistent with the hypotheses $(A.1)-(A.3)$. The global existence, uniqueness, and asymptotic behaviour of solutions are proved in the appendix. The theoretical results in Section 3 are true for non-linear age-dependent models with aging and birth functions satisfying (H.1) and (H.2), where (H.2) applies to different conditions than that in (\citealt{webbbook}). Compared to previous models with infection age as a continuous variable, we are able to incorporate some important aspects of the spread and control of an epidemic disease in the model, together with practical interpretations of the corresponding parameters. For example, (i) by considering the simplified fact that the tracing rate varies according to case identification rate, we will be able to understand one of the reasons for small fluctuations that usually appear in daily cases in many of the real data; (ii) decrease of susceptible population due to public health interventions is not negligible when the interventions successfully protect a considerable amount of people from infection.\\

In application I, we use our model to assess public health guidelines in the event of a smallpox bioterrorist attack in a large urban center. Our simulation falls into the scenario that releases of the virus take place in the community with people being unaware of them. But we can easily modify the initial conditions to simulate other initial scenarios, such as when index cases are introduced into the community by a smallpox release in another area, while the government and the public are getting prepared and in a watchful state. Our simulation results point out that with a limited amount of vaccine stockpiles and healthcare workers, ring vaccination is more efficient in preventing the disease from spreading than mass vaccination. With the initial condition in \emph{Fig.} \ref{initial}, there are not many people in the vaccination queue at the beginning of the outbreak. In this case, ring vaccination allows the early vaccine distribution for selected groups to enhance the response readiness (\citealt{CDCresponse}), and hence allows more efficient utilization of vaccination capacity.\\

We also investigate the ring vaccination effectiveness by varying the three key parameters: isolation rate, contacts traced per case, and contact tracing efficiency in finding infectives. \emph{Fig.} \ref{Rsymeta} and \emph{Fig.} \ref{RsymCT} also confirm the conclusion in (\citealt{Day}): tracing and quarantine help avert more cases when the isolation of symptomatic cases is ineffective. Additionally, we show that in the case of smallpox, the effectiveness of ring vaccination in reducing infections increases \emph{at an accelerating rate as the effectiveness of isolation diminishes\footnote{Quote from (\citealt{Day})}} when the ring vaccination efficacy is of a normal level, but it increases at an almost constant lower rate when the ring vaccination efficacy is of a higher level.\\

Our model is able to provide guidance to public health decisions to adjust current contact tracing strategies either before or during an outbreak with updated data. All the parameters in our model have good epidemiological interpretations and are easy to estimate with data from historical epidemic outbreaks. Unlike many other studies, we take into consideration susceptible population variation due to quarantine and vaccination, which usually leads to community herd immunity. So when vaccines are available, our model can be applied to provide guidelines for vaccination strategies to create herd immunity\footnote{A historical example is the "ring" vaccination strategy used to eliminate smallpox}.\\

In application II, we show that our simulation of SARS in Taiwan fits well with the observed data, and we are able to answer the question about how many cases are avoided by implementing contact tracing and quarantine in the control of the outbreak. With more precise data about each identified case, we would be able to estimate the case isolation rate accurately, and our model would enable us to determine the length of incubation period by data fitting. Therefore, our model would also be helpful in estimating important parameters and predicting transmission dynamics with evolving data during an outbreak. A great difference between the two applications we present in the paper is, contact tracing applied to control SARS in Taiwan, 2003 is not as effective as ring vaccination strategy in eradicating smallpox. The theoretical reason is that we have different settings of contact tracing parameters $CT$, ${\eta_I}$, and ${\eta_S}$ in the two applications. In reality, our settings are quite reasonable due to facts such as the severity of symptoms, availability of vaccines (since vaccination is an important way to create herd immunity), and readiness of the public health officials with the preparedness and containment plans.\\

Our simulations can guide public health officials in adjusting levels of different strategies to control the outbreak and deploy resources efficiently. With theoretical suggestions, realistic adjustments about how to deploy limited resources (such as vaccine stockpiles, healthcare workers, surveillance stations, \emph{etc.}) to meet the theoretical levels would strongly depend on the decisions of public health authorities. Furthermore, with cost-effectiveness data, we can apply optimal control methods to quantitatively determine the best control strategies.\\

Furthermore, the age-structured model possesses great potential in modeling the vaccination strategy of HIV (although the HIV vaccine does not exist so far, several encouraging studies such as (\citealt{hiv}) suggest that there is a significant hope in the future). The HIV infection has an extremely long asymptomatic period and many HIV-positive people are unaware of their infection with the virus. Thus, an active infected individual would spread the disease without even being aware of the infection, which makes the control and detection of HIV very difficult. Even though we might have an HIV vaccine available in the future, with possible serious side-effects, it might be too limited and costly to be available to everyone at the beginning. So the deployment of a limited amount of vaccine will be a serious issue. Then, because of the prolonged asymptomatic stage of HIV infection, the effects of certain intervention strategies would depend even more on the age of infection. Our model will be an advantageous starting point for such investigation.

%%%%%%%%%%%%%%%%%%%%%%%%%%%%%%%%%%%%%%%%%%%%%%%%%%%%%%%%%%%%%%%%%%%%%%%%%%%%%%%%%%%%%%%%%%%%%%%%%%%%%%%%%%%%%%%%%%%%%%%%%%%%%
%%%%%%%%%%%%%%%%%%%%%%%%%%%%%%%%%%%%%%%%%%%%%%%%%%%%%%%%%%%%%%%%%%%%%%%%%%%%%%%%%%%%%%%%%%%%%%%%%%%%%%%%%%%%%%%%%%%%%%%%%%%%%%
%%%%%%%%%%%%%%%%%%%%%%%%%%%%%%%%%%%%%%%%%%%%%%%%%%%%%%%%%%%%%%%%%%%%%%%%%%%%%%%%%%%%%%%%%%%%%%%%%%%%%%%%%%%%%%%%%%%%%%%%%%%%%%%%
%%%%%%%%%%%%%%%%%%%%%%%%%%%%%%%%%%%%%%%%%%%%%%%%%%%%%%%%%%%%%%%%%%%%%%%%%%%%%%%%%%%%%%%%%%%%%%%%%%%%%%%%%%%%%%%%%%%%%%%%%%%%%%%

\begin{acknowledgements}
The author would like to express her sincere appreciation to her Ph.D. advisor, Professor Glenn Webb, for his patient guidance and constant help throughout this research.
\end{acknowledgements}

\bibliographystyle{spbasic}      % basic style, author-year citations

\bibliography{references}   % name your BibTeX data base

%\begin{thebibliography}{}
% etc
%\end{thebibliography}
\section{Appendix}
%%%%%%%%%%%%%%%%%%%%%%%%%%%%%%%%%%%%%%%%%%%%%%%%%%%%%%%%%%%%%%%%%%%%%%%%%%%%%%%%%%%%%%%%%%%%%%%%%%%%%%%%%%%%%%%%%%%%%%%
%%%%%%%%%%%%%%%%%%%%%%%%%%%%%%%%%%%%%%%%%%%%%%%%%%%%%%%%%%%%%%%%%%%%%%%%%%%%%%%%%%%%%%%%%%%%%%%%%%%%%%%%%%%%%%%%%%%%%%%
Theorem \ref{local existence and uniqueness} can be proved by the following three propositions:
\begin{prop}\label{equivalence}
Let (H.1), (H.2) hold, let $T > 0$, let $\phi  \in {L^1}$, and let $l \in {C_T}$. If $l$ is a solution of the integral equation:
%%%%%%%%%%%%%%%%%%%%%%%%%%%%%%%%%%%%%%%%%%%%%%%%%%%%%%%%%%%%%%%%%%%%%%%%%%%%%%%%%%%%%%%%%%%%%%%%%%%%%%%%%%%%%%%%%%%%%%%%%%%%%%%%%%%%%%%%%%%%%%%%%%%%%%%%%%%%%%%%%%%%%%%%%%%%%%%%%%%%%%%%%%%%%%%%%%%%%%%%%%%%%%%%%%%%%%%%%%%%%%%%%%%%%%%%%%%%%%%%%%%%%%%%%%%%%%%%
\begin{equation}\label{equivalent integral ADP}
l\left( t \right)\left( a \right) = \left\{ \begin{matrix} 
  \left( {F\left( l \right)} \right)\left( {t - a} \right) + \int_{t - a}^t {G\left( {l\left( s \right)} \right)\left( {s + a - t} \right)ds} ,\,0 < a < t \hfill \cr 
  \phi \left( {a - t} \right) + \int_0^t {G\left( {l\left( s \right)} \right)\left( {s + a - t} \right)ds} ,\,t \le a \le \emph{M} \hfill \cr 
 \end{matrix}  \right.
\end{equation}
on $[0,T]$, then $l$ is a solution of \eqref{ADP} on $[0,T]$.
\end{prop}
\begin{proof}
The proof is similar to that of Proposition $2.1$ in (\citealt{webbbook}), except that we can use the uniform continuity of the function $t \mapsto \left( {F\left( l \right)} \right)\left( t \right)$ from $\left[ {0,T} \right]$ to $\mathbb{R}$ for $l \in {C_T}$ instead in this proof.
\end{proof}
%%%%%%%%%%%%%%%%%%%%%%%%%%%%%%%%%%%%%%%%%%%%%%%%%%%%%%%%%%%%%%%%%%%%%%%%%%%%%%%%%%%%%%%%%%%%%%%%%%%%%%%%%%%%%%%%%%%%%%%%%%%%%%%%%%%%%%%%%%%%%%%%%%%%%%%%%%%%%%%%%%%%%%%%%%%%%%%%%%%%%%%%%%%%%%%%%%%%%%%%%%%%%%%%%%%%%%%%%%%%%%%%%%%%%%%%%%%%%%%%%%%%%%%%%%%%%%%%
\begin{prop}\label{existence and uniqueness of integral ADP solution}
Let (H.1), (H.2) hold and let $r>0$. There exists $T>0$ such that if $\phi  \in {L^1}$ and ${\left\| \phi  \right\|_{{L^1}}} \le r$, then there is a unique function $l \in {C_T}$ such that $l$ is a solution of \eqref{equivalent integral ADP} on $[0,T]$.
\end{prop}
\begin{proof}
We will prove it by contraction mapping theorem. We fix $r \geqslant {\left\| \phi  \right\|_{{L^1}}} > 0$ and choose $T>0$ such that 
$$T \cdot \left[ {{c_1}\left( {2r} \right) + {c_2}\left( {2r,T} \right) + \frac{{\mathop {\sup }\limits_{0 \le t \le T} \left| {\left( {F\left( 0 \right)} \right)\left( t \right)} \right| + {{\left\| {G\left( 0 \right)} \right\|}_{{L^1}}}}}{{2r}}} \right] \le \frac{1}{2}$$
Then define $S$ as a closed subset of ${C_T}$: 
$$S: = \left\{ {l \in {C_T}:l\left( 0 \right) = \phi ,\,{{\left\| l \right\|}_{{C_T}}} \le 2r} \right\}$$
We define a mapping $K$ on $S$ as following and prove that $K$ is a strict contraction from $S$ into $S$.
\begin{equation*}
\left( {K\left( l \right)} \right)\left( t \right)\left( a \right) = \left\{ \begin{matrix} 
  \left( {F\left( l \right)} \right)\left( {t - a} \right) + \int_{t - a}^t {G\left( {l\left( s \right)} \right)\left( {s + a - t} \right)ds} ,\;a.e.\,a \in \left( {0,t} \right) \hfill \cr 
  \phi \left( {a - t} \right) + \int_0^t {G\left( {l\left( s \right)} \right)\left( {s + a - t} \right)ds} ,\,a.e.\,a \in \left[ {t,\emph{M}} \right] \hfill \cr 
 \end{matrix}  \right.
\end{equation*}
we need to verify that the following conditions hold:
\begin{enumerate}[(i)]
\item
Let $l\in S$, $t \in \left[ {0,T} \right]$, then ${\left\| {\left( {K\left( l \right)} \right)\left( t \right)} \right\|_{{L^1}}} \leqslant 2r$.
\item
Let $l\in S$ and let $0 \le t < \hat t \le T$, then ${\left\| {\left( {K\left( l \right)} \right)\left( t \right) - \left( {K\left( l \right)} \right)\left( {\hat t} \right)} \right\|_{{L^1}}} \to 0$ as $\left| {\hat t - t} \right| \to 0$.
\item
Let ${l_1}$, ${l_2} \in S$, then ${\left\| {K\left( {{l_1}} \right) - K\left( {{l_2}} \right)} \right\|_{{C_T}}} \le \frac{1}{2}{\left\| {{l_1} - {l_2}} \right\|_{{C_T}}}$.
\end{enumerate}
For (i), we will only consider the case when $0 \le t \le \min \{ \emph{M},T\} $. (Otherwise, we have $t \geqslant \emph{M} \geqslant a$, then we just need to consider the expression of $\left( {K\left( l \right)} \right)\left( t \right)\left( a \right)$ for $a \in \left( {0,t} \right)$.)
\begin{align*}
& {\left\| {\left( {K\left( l \right)} \right)\left( t \right)} \right\|_{{L^1}}} = \int_0^\emph{M} {\left| {\left( {K\left( l \right)} \right)\left( t \right)\left( a \right)} \right|da} \\
 \leqslant & \int_0^t {\left| {\left( {F\left( l \right)} \right)\left( {t - a} \right) + \int_{t - a}^t {G\left( {l\left( s \right)} \right)\left( {s + a - t} \right)ds} } \right|da} \\
& + \int_t^\emph{M} {\left| {\phi \left( {a - t} \right) + \int_0^t {G\left( {l\left( s \right)} \right)\left( {s + a - t} \right)ds} } \right|da} \\
\leqslant & \int_0^t {\left| {\left( {F\left( l \right)} \right)\left( s \right)} \right|ds}  + \int_0^{\emph{M} - t} {\left| {\phi \left( s \right)} \right|ds}  + \int_0^t {\int_{t - s}^\emph{M} {\left| {G\left( {l\left( s \right)} \right)\left( {s + a - t} \right)} \right|dads} } \\
\leqslant & \int_0^t {\left| {\left( {F\left( l \right)} \right)\left( s \right) - \left( {F\left( 0 \right)} \right)\left( s \right)} \right|ds}  + \int_0^t {\left| {\left( {F\left( 0 \right)} \right)\left( s \right)} \right|ds}  + \int_0^\emph{M} {\left| {\phi \left( s \right)} \right|ds}\\
 & + \int_0^t {{{\left\| {G\left( {l\left( s \right)} \right) - G\left( 0 \right)} \right\|}_{{L^1}}}ds}  + \int_0^t {{{\left\| {G\left( 0 \right)} \right\|}_{{L^1}}}ds} \\
 \leqslant & {c_2}\left( {2r,t} \right)\int_0^t {\mathop {\sup }\limits_{0 \leqslant \tau  \leqslant s} {{\left\| {l\left( \tau  \right)} \right\|}_{{L^1}}}ds}  + t \cdot \mathop {\sup }\limits_{0 \leqslant s \leqslant t} \left| {\left( {F\left( 0 \right)} \right)\left( s \right)} \right| + {\left\| \phi  \right\|_{{L^1}}}\\
 & + {c_1}\left( {2r} \right)\int_0^t {{{\left\| {l\left( s \right)} \right\|}_{{L^1}}}ds}  + t \cdot {\left\| {G\left( 0 \right)} \right\|_{{L^1}}}\\
 \le & 2rT \left[ {{c_1}\left( {2r} \right) + {c_2}\left( {2r,T} \right) + \frac{{\mathop {\sup }\limits_{0 \le t \le T} \left| {\left( {F\left( 0 \right)} \right)\left( t \right)} \right| + {{\left\| {G\left( 0 \right)} \right\|}_{{L^1}}}}}{{2r}}} \right] + r
 \le  2r
\end{align*}
For (ii), we can just follow the same estimation in the proof of Proposition $2.2$ in (\citealt{webbbook}), except that we need to use the uniform continuity of the function $t \mapsto \left( {F\left( l \right)} \right)\left( t \right)$ from $\left[ {0,T} \right]$ to $\mathbb{R}$ for $l \in {C_T}$. (i) and (ii) imply that $K$ maps $S$ into $S$, (iii) shows that $K$ is a contraction.\\
To prove (iii), given any ${l_1}$, ${l_2} \in S$, we consider $0 \le t \le \min \{ \emph{M},T\} $. Similarly we have:
\begin{align*}
&\int_0^\emph{M} {\left| {\left( {K\left( {{l_1}} \right)} \right)\left( t \right)\left( a \right) - \left( {K\left( {{l_2}} \right)} \right)\left( t \right)\left( a \right)} \right|da} \\
\le & \int_0^t {\left| {\left( {F\left( {{l_1}} \right)} \right)\left( s \right) - \left( {F\left( {{l_2}} \right)} \right)\left( s \right)} \right|ds}  + \int_0^t {{{\left\|{G\left( {{l_1}\left( s \right)} \right) - G\left( {{l_2}\left( s \right)} \right)} \right\|}_{{L^1}}}ds} \\
\le & \int_0^t {{c_2}\left( {2r,s} \right)\mathop {\sup }\limits_{0 \leqslant \tau  \leqslant s} {{\left\| {{l_1}\left( \tau  \right) - {l_2}\left( \tau  \right)} \right\|}_{{L^1}}}ds}  + {c_1}\left( {2r} \right)\int_0^t {{{\left\| {{l_1}\left( s \right) - {l_2}\left( s \right)} \right\|}_{{L^1}}}ds} \\
 \le & T \cdot \left[ {{c_1}\left( {2r} \right) + {c_2}\left( {2r,T} \right)} \right]{\left\| {{l_1} - {l_2}} \right\|_{{C_T}}}
 \le  \frac{1}{2}{\left\| {{l_1} - {l_2}} \right\|_{{C_T}}}
\end{align*}
\end{proof}

%%%%%%%%%%%%%%%%%%%%%%%%%%%%%%%%%%%%%%%%%%%%%%%%%%%%%%%%%%%%%%%%%%%%%%%%%%%%%%%%%%%%%%%%%%%%%%%%%%%%%%%%%%%%%%%%%%%%%%%%%%%%%%
%%%%%%%%%%%%%%%%%%%%%%%%%%%%%%%%%%%%%%%%%%%%%%%%%%%%%%%%%%%%%%%%%%%%%%%%%%%%%%%%%%%%%%%%%%%%%%%%%%%%%%%%%%%%%%%%%%%%%%%%%%%%%%
\begin{prop}\label{uniqueness of ADP solution}
Let (H.1), (H.2) hold, let $\phi ,\hat \phi  \in {L^1}$, let $T>0$, and let $l,\hat l \in {C_T}$ such that $l$, $\hat l$ is the solution of \eqref{ADP} on $[0,T]$ for $\phi$, $\hat \phi$, respectively. Let $r>0$ such that ${\left\| l \right\|_{{C_T}}},{\left\| {\hat l} \right\|_{{C_T}}} \le r$. Then for $0 \le t \le T$,
$${\left\| {l\left( t \right) - \hat l\left( t \right)} \right\|_{{L^1}}} \leqslant {e^{\left[ {{c_1}\left( r \right) + {c_2}\left( {r,T} \right)} \right]t}}{\left\| {\phi  - \hat \phi } \right\|_{{L^1}}}$$
Hence we have the uniqueness of the local solution of \eqref{ADP}.
\end{prop}
\begin{proof}
For each $t \in \left[ {0,T} \right]$ we define two continuous functions:
\begin{enumerate}[({}1)]
\item
$V\left( t \right): = {\left\| {l\left( t \right) - \hat l\left( t \right)} \right\|_{{L^1}}} = \int_{ - t}^{\emph{M} - t} {\left| {l\left( t \right)\left( {t + c} \right) - \hat l\left( t \right)\left( {t + c} \right)} \right|dc}$
\item
$W\left( t \right): = \mathop {\sup }\limits_{0 \leqslant s \leqslant t} {\left\| {l\left( s \right) - \hat l\left( s \right)} \right\|_{{L^1}}} = \mathop {\sup }\limits_{0 \leqslant s \leqslant t} V\left( s \right)$
\end{enumerate}
Next, we estimate $\mathop {\lim \sup }\limits_{h \to {0^ + }} {h^{ - 1}}\left[ {W\left( {t + h} \right) - W\left( t \right)} \right]$ for each fixed $t \in \left[ {0,T} \right]$ separately under the following two situations:
\begin{enumerate}[(i)]
\item
$V\left( t \right) < W\left( t \right)$ (as shown in \emph{Fig.} \ref{proof}), i.e., $\exists {t_0} < t$ such that $W\left( t \right) = V\left( {{t_0}} \right) > V\left( t \right)$. Since the mapping $s \mapsto V\left( s \right)$ is continuous, we can choose sufficiently small $h  > 0$ such that $V\left( {t + \delta } \right) \leqslant V\left( {{t_0}} \right)$ for $0 \leqslant \delta  \leqslant h$, hence $W\left( {t + \delta } \right) = W\left( t \right)$ for $0 \leqslant \delta  \leqslant h$. Then $\mathop {\lim \sup }\limits_{h \to {0^ + }} {h^{ - 1}}\left[ {W\left( {t + h} \right) - W\left( t \right)} \right] = 0$.
\begin{figure}[htbp]
\begin{center}
\includegraphics[scale=0.4]{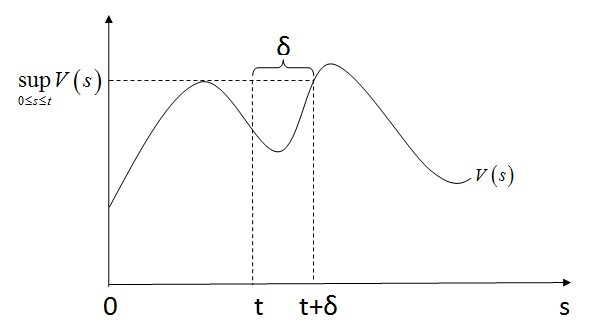}
\caption{\label{proof}}
\end{center}
\end{figure}
\item
$V\left( t \right) = W\left( t \right)$, i.e.,\ the function $V$ attains the supremum value in $\left[ {0,t} \right]$ at $t$. Then we have
\begin{flalign*}
&{h^{ - 1}}\left[ {W\left( {t + h} \right) - W\left( t \right)} \right]
=  {h^{ - 1}}\left[ {\mathop {\sup }\limits_{0 \leqslant s \leqslant t + h} V\left( s \right) - V\left( t \right)} \right]\\
= & {h^{ - 1}}\left[ {\max \left\{ {V\left( t \right),\mathop {\sup }\limits_{t \le s \le t + h} V\left( s \right)} \right\} - V\left( t \right)} \right]
\le  {h^{ - 1}}\left| {\mathop {\sup }\limits_{t \le s \le t + h} V\left( s \right) - V\left( t \right)} \right|\\
\le & {h^{ - 1}}\mathop {\sup }\limits_{t \leqslant s \leqslant t + h} \left| {V\left( s \right) - V\left( t \right)} \right|
=  \mathop {\sup }\limits_{0 \leqslant {h_0} \leqslant h} \frac{{{h_0}}}{h}h_0^{ - 1}\left| {V\left( {t + {h_0}} \right) - V\left( t \right)} \right|\\
\leqslant & \mathop {\sup }\limits_{0 \leqslant {h_0} \leqslant h} h_0^{ - 1}\left| {V\left( {t + {h_0}} \right) - V\left( t \right)} \right|
\end{flalign*}
For each ${h_0} \in \left[ {0,h} \right]$, we estimate
\begin{flalign*}
&h_0^{ - 1}\left[ {V\left( {t + {h_0}} \right) - V\left( t \right)} \right]\\
= & h_0^{ - 1}\int_{ - t - {h_0}}^{ - t} {\left| {l\left( {t + {h_0}} \right)\left( {t + {h_0} + c} \right) - \hat l\left( {t + {h_0}} \right)\left( {t + {h_0} + c} \right)} \right|dc} \\
& + h_0^{ - 1}\int_{ - t}^{\emph{M} - t - {h_0}} {\left| {l\left( {t + {h_0}} \right)\left( {t + {h_0} + c} \right) - \hat l\left( {t + {h_0}} \right)\left( {t + {h_0} + c} \right)} \right|dc} \\
& - h_0^{ - 1}\int_{ - t}^{\emph{M} - t} {\left| {l\left( t \right)\left( {t + c} \right) - \hat l\left( t \right)\left( {t + c} \right)} \right|dc} \\
\leqslant & h_0^{ - 1}\int_0^{{h_0}} {\left| {l\left( {t + {h_0}} \right)\left( a \right) - \left( {F\left( l \right)} \right)\left( t \right)} \right|da} \\
& + h_0^{ - 1}\int_0^{{h_0}} {\left| {\left( {F\left( l \right)} \right)\left( t \right) - \left( {F\left( {\hat l} \right)} \right)\left( t \right)} \right|da} \\
& + h_0^{ - 1}\int_0^{{h_0}} {\left| {\left( {F\left( {\hat l} \right)} \right)\left( t \right) - \hat l\left( {t + {h_0}} \right)\left( a \right)} \right|da} \\
& + \int_0^\emph{M} {\left| {h_0^{ - 1}\left[ {l\left( {t + {h_0}} \right)\left( {a + {h_0}} \right) - l\left( t \right)\left( a \right)} \right] - G\left( {l\left( t \right)} \right)\left( a \right)} \right|da} \\
& + \int_0^\emph{M} {\left| {G\left( {l\left( t \right)} \right)\left( a \right) - G\left( {\hat l\left( t \right)} \right)\left( a \right)} \right|da} \\
& + \int_0^\emph{M} {\left| {h_0^{ - 1}\left[ {\hat l\left( {t + {h_0}} \right)\left( {a + {h_0}} \right) - \hat l\left( t \right)\left( a \right)} \right] - G\left( {\hat l\left( t \right)} \right)\left( a \right)} \right|da} 
\end{flalign*}
\end{enumerate}
So for both of the above situations, notice that both $l$ and $\hat l$ are solutions of problem \eqref{ADP}, we can estimate as the following:
\begin{align*}
&\mathop {\lim \sup }\limits_{h \to {0^ + }} {h^{ - 1}}\left[ {W\left( {t + h} \right) - W\left( t \right)} \right]
 \leqslant  \mathop {\lim \sup }\limits_{h \to {0^ + }} \mathop {\sup }\limits_{0 \leqslant {h_0} \leqslant h} h_0^{ - 1}\left| {V\left( {t + {h_0}} \right) - V\left( t \right)} \right|\\
 \leqslant & \left| {\left( {F\left( l \right)} \right)\left( t \right) - \left( {F\left( {\hat l} \right)} \right)\left( t \right)} \right| + {\left\| {G\left( {l\left( t \right)} \right) - G\left( {\hat l\left( t \right)} \right)} \right\|_{{L^1}}}\\
 \leqslant & {c_2}\left( {r,t} \right)\mathop {\sup }\limits_{0 \leqslant s \leqslant t} {\left\| {l\left( s \right) - \hat l\left( s \right)} \right\|_{{L^1}}} + {c_1}\left( r \right){\left\| {l\left( t \right) - \hat l\left( t \right)} \right\|_{{L^1}}}\\
 \leqslant & \left[ {{c_1}\left( r \right) + {c_2}\left( {r,T} \right)} \right]W\left( t \right)
\end{align*}
So we have $W\left( t \right) \leqslant {e^{\left[ {{c_1}\left( r \right) + {c_2}\left( {r,T} \right)} \right]t}}W\left( 0 \right)$ (\citealt{book}, Theorem 1.4.1). Hence, 
$$V\left( t \right) \leqslant W\left( t \right) \leqslant {e^{\left[ {{c_1}\left( r \right) + {c_2}\left( {r,T} \right)} \right]t}}W\left( 0 \right) = {e^{\left[ {{c_1}\left( r \right) + {c_2}\left( {r,T} \right)} \right]t}}{\left\| {\phi  - \hat \phi } \right\|_{{L^1}}}$$
That is,
$${\left\| {l\left( t \right) - \hat l\left( t \right)} \right\|_{{L^1}}} \leqslant {e^{\left[ {{c_1}\left( r \right) + {c_2}\left( {r,T} \right)} \right]t}}{\left\| {\phi  - \hat \phi } \right\|_{{L^1}}}$$
\end{proof}

%%%%%%%%%%%%%%%%%%%%%%%%%%%%%%%%%%%%%%%%%%%%%%%%%%%%%%%%%%%%%%%%%%%%%%%%%%%%%%%%%%%%%%%%%%%%%%%%%%%%%%%%%%%%%%%%%%%%%%%%%%%%%%
%%%%%%%%%%%%%%%%%%%%%%%%%%%%%%%%%%%%%%%%%%%%%%%%%%%%%%%%%%%%%%%%%%%%%%%%%%%%%%%%%%%%%%%%%%%%%%%%%%%%%%%%%%%%%%%%%%%%%%%%%%%%%%
\begin{proof} [Proof of Theorem \ref{existence and uniqueness of positive solution} and Theorem \ref{global solution}]
The proof of Theorem \ref{existence and uniqueness of positive solution} and Theorem \ref{global solution} are similar to that of the corresponding theorems in Sections $2.3-2.4$ in (\citealt{webbbook}). We only need to switch the statement of Proposition $2.4$ in (\citealt{webbbook}) to the following Proposition \ref{semigroup property}.
\end{proof}
%%%%%%%%%%%%%%%%%%%%%%%%%%%%%%%%%%%%%%%%%%%%%%%%%%%%%%%%%%%%%%%%%%%%%%%%%%%%%%%%%%%%%%%%%%%%%%%%%%%%%%%%%%%%%%%%%%%%%%%%%%%%%%%%%%%%%%%%%%%%%%%%%%%%%%%%%%%%%%%%%%%%%%%%%%%%%%%%%%%%%%%%%%%%%%%%%%%%%%%%%%%%%%%%%%%%%%%%%%%%%%%%%%%%%%%%%%%%%%%%%%%%%%%%%%%%%%%%
\begin{prop}\label{semigroup property}
Let (H.1), (H.2) hold, let $\phi \in {L^1}$, let $T>0$, and let $l \in {C_T}$ such that $l$ is a solution of \eqref{equivalent integral ADP} on ${\left[ {0,T} \right]}$. Let $\hat T > 0$ and let $\hat l \in {C_{T + T}}$ such that $\hat l\left( t \right) = l\left( t \right)$ for $t \in \left[ {0,T} \right]$, and for $t \in \left( {T,T + \hat T} \right]$, $\hat l$ satisfies the following integral equation:
$$\hat l\left( t \right)\left( a \right) = \left\{ \begin{matrix} 
  \left( {F\left( {\hat l} \right)} \right)\left( {t - a} \right) + \int_{t - a}^t {G\left( {\hat l\left( s \right)} \right)\left( {s + a - t} \right)ds} ,\;0 < a < t - T \hfill \cr 
  l\left( T \right)\left( {a - t + T} \right) + \int_T^t {G\left( {\hat l\left( s \right)} \right)\left( {s + a - t} \right)ds} ,\;t - T \leqslant a \leqslant \emph{M} \hfill \cr 
 \end{matrix}  \right.$$
Then $\hat l$ is a solution of \eqref{equivalent integral ADP} on $\left[ {0,T + \hat T} \right]$.
\end{prop}
\begin{proof}
First of all, we notice that if $(H.2)$ holds, then $\left( {F\left( l \right)} \right)\left( t \right) = \left( {F\left( {\hat l} \right)} \right)\left( t \right)$ for $t \in \left[ {0,T} \right]$. Because for any $t \in \left[ {0,T} \right]$, by $(H.2)$, $\exists C > 0$ such that
$$\left| {\left( {F\left( l \right)} \right)\left( t \right) - \left( {F\left( {\hat l} \right)} \right)\left( t \right)} \right| \leqslant C\mathop {\sup }\limits_{0 \leqslant s \leqslant t} {\left\| {l\left( s \right) - \hat l\left( s \right)} \right\|_{{L^1}}} = 0$$
Then it is easy to verify that $\hat{l}$ is a solution to \eqref{equivalent integral ADP} for $t \in \left[ {0,T} \right]$. Next, we verify that $\hat{l}$ is a solution to \eqref{equivalent integral ADP} for $t \in \left( {T,T + \hat T} \right]$.
\\For $t - T \leqslant a \leqslant \emph{M}$:
\begin{enumerate}
\item
If $a \geqslant t$,
\begin{flalign*}
&\hat l\left( t \right)\left( a \right) = l\left( T \right)\left( {a - t + T} \right) + \int_T^t {G\left( {\hat l\left( s \right)} \right)\left( {s + a - t} \right)ds}\\
& = \left( {F\left( l \right)} \right)\left( {t - a} \right) + \int_{t - a}^T {G\left( {l\left( s \right)} \right)\left( {s + a - t} \right)ds}  + \int_T^t {G\left( {\hat l\left( s \right)} \right)\left( {s + a - t} \right)ds} \\
& = \left( {F\left( {\hat l} \right)} \right)\left( {t - a} \right) + \int_{t - a}^t {G\left( {\hat l\left( s \right)} \right)\left( {s + a - t} \right)ds} 
\end{flalign*}
\item
If $a < t$,
\begin{flalign*}
&\hat l\left( t \right)\left( a \right) = l\left( T \right)\left( {a - t + T} \right) + \int_T^t {G\left( {\hat l\left( s \right)} \right)\left( {s + a - t} \right)ds} \\
& = \phi \left( {a - t} \right) + \int_0^T {G\left( {l\left( s \right)} \right)\left( {s + a - t} \right)ds}  + \int_T^t {G\left( {\hat l\left( s \right)} \right)\left( {s + a - t} \right)ds} \\
& = \phi \left( {a - t} \right) + \int_0^t {G\left( {\hat l\left( s \right)} \right)\left( {s + a - t} \right)ds} 
\end{flalign*}
\end{enumerate}
For $0 \leqslant a \leqslant t-T$, the verification is straightforward.
\end{proof}

%%%%%%%%%%%%%%%%%%%%%%%%%%%%%%%%%%%%%%%%%%%%%%%%%%%%%%%%%%%%%%%%%%%%%%%%%%%%%%%%%%%%%%%%%%%%%%%%%%%%%%%%%%%%%%%%%%%%%%%%%%%%%%%%%%%%%%%%%%%%%%%%%%%%%%%%%%%%%%%%%%%%%%%%%%%%%%%%%%%%%%%%%%%%%%%%%%%%%%%%%%%%%%%%%%%%%%%%%%%%%%%%%%%%%%%%%%%%%%%%%%%%%%%%%%%%%%%%
\begin{proof}[Proof of Theorem \ref{existence and uniqueness of main problem}]
By Theorem \ref{existence and uniqueness of positive solution} and Theorem \ref{global solution}, We only need to show that the aging function $G$ in $(P.1)$ and the birth function $F$ in $(P.2)$ satisfy the hypotheses $(H.1)-(H.5)$. $(H.3)$ and $(H.4)$ are obvious from $(P.1)$ and $(P.2)$, we will justify $(H.1)$ and $(H.2)$ as follows. First denote $\bar \mu : = {\left\| \mu  \right\|_{{L^\infty }}}$. Let $r>0$, for any ${\phi _1}$, ${\phi _2} \in L^1$ such that ${\left\| {{\phi _1}} \right\|_{{L^1}}}$, ${\left\| {{\phi _2}} \right\|_{{L^1}}} \le r$, we have
\begin{flalign*}
\begin{split}
&{\left\| {G\left( {{\phi _1}} \right) - G\left( {{\phi _2}} \right)} \right\|_{{L^1}}}\\
\leqslant & \int_0^\emph{M} {\mu \left( a \right)\left| {{\phi _1}\left( a \right) - {\phi _2}\left( a \right)} \right|da}  + \mathcal{T}\left( {{\phi _1}} \right)\int_0^\emph{M} {\left| {{\phi _1}\left( a \right) - {\phi _2}\left( a \right)} \right|da} \\
& + \left| {\mathcal{T}\left( {{\phi _1}} \right) - \mathcal{T}\left( {{\phi _2}} \right)} \right|\int_0^\emph{M} {\left| {{\phi _2}\left( a \right)} \right|da} \\
 \leqslant & \left( {\bar \mu}  + 2r{\left\| \mathcal{T} \right\|_\infty } \right){\left\| {{\phi _1} - {\phi _2}} \right\|_{{L^1}}}
\end{split}
\end{flalign*}
so we have $(H.1)$. In order to prove $(H.2)$,  first we notice that $t \mapsto \left( {F\left( \phi  \right)} \right)\left( t \right)$ is continuous in $\left[ {0,\infty } \right)$ for any $\phi  \in {C_{T}}$. So $F:{C_{T}} \to C\left( {\left[ {0,T} \right];{\mathbb{R}}} \right)$.\\
Next, for any ${\phi _1}$, ${\phi _2} \in {C_{T}}$ such that ${\left\| {{\phi _1}} \right\|_{{C_T}}}$, ${\left\| {{\phi _2}} \right\|_{{C_T}}} \le r$, we have $\forall t \in \left[ {0,T} \right]$,
\begin{align*}
&\left| {\left( {F\left( {{\phi _1}} \right)} \right)\left( t \right) - \left( {F\left( {{\phi _2}} \right)} \right)\left( t \right)} \right|\\
\le & {S_0}\left| {\mathcal{B}\left( {{\phi _1}\left( t \right)} \right) - \mathcal{B}\left( {{\phi _2}\left( t \right)} \right)} \right|{e^{ - \int_0^t {\mathcal{B} \left( {{\phi _1}\left( s \right)} \right) + \mathcal{Q}\left( {{\phi _1}\left( s \right)} \right)ds} }}\\
& + {S_0}\mathcal{B} \left( {{\phi _2}\left( t \right)} \right)\left| {{e^{ - \int_0^t {\mathcal{B} \left( {{\phi _1}\left( s \right)} \right) + \mathcal{Q}\left( {{\phi _1}\left( s \right)} \right)ds} }} - {e^{ - \int_0^t {\mathcal{B} \left( {{\phi _2}\left( s \right)} \right) + \mathcal{Q}\left( {{\phi _2}\left( s \right)} \right)ds} }}} \right|\\
 = & :{I_1} + {I_2}
\end{align*}
Obviously, ${I_1} \le {S_0}{\left| \mathcal{B} \right|} {\left\| {{\phi _1}\left( t \right) - {\phi _2}\left( t \right)} \right\|_{{L^1}}}{e^{rt\left( {\left| \mathcal{Q} \right| + \left| \mathcal{B} \right|} \right)}}$.\\
In order to consider ${I_2}$, notice that $\left| {{e^X} - {e^Y}} \right| \le {e^\emph{M}}\left| {X - Y} \right|$ for $\left| X \right|$, $\left| Y \right| \le \emph{M}$. Then we have
\begin{align*}
{I_2} & \leqslant {S_0}{\left| \mathcal{B} \right|} r {e^{rt\left( {\left| \mathcal{Q} \right| + \left| \mathcal{B} \right|} \right)}} \int_0^t {\left| {\mathcal{B} \left( {{\phi _1}\left( s \right)} \right) + \mathcal{Q}\left( {{\phi _1}\left( s \right)} \right) - \mathcal{B} \left( {{\phi _2}\left( s \right)} \right) - \mathcal{Q}\left( {{\phi _2}\left( s \right)} \right)} \right|ds} \\
\le & {S_0} {\left| \mathcal{B} \right|} r {e^{rt\left( {\left| \mathcal{Q} \right| + \left| \mathcal{B} \right|} \right)}} \left( {\left| \mathcal{B} \right|} +{\left| \mathcal{Q} \right|} \right)\int_0^t {{{\left\| {{\phi _1}\left( s \right) - {\phi _2}\left( s \right)} \right\|}_{{L^1}}}ds} 
\end{align*}
Hence,
\begin{flalign*}
& \left| {\left( {F\left( {{\phi _1}} \right)} \right)\left( t \right) - \left( {F\left( {{\phi _2}} \right)} \right)\left( t \right)} \right|\\
\le & {S_0}\left| \mathcal{B} \right|{e^{rt\left( {\left| \mathcal{Q} \right| + \left| \mathcal{B} \right|} \right)}}\left[ {1 + rt\left( {\left| \mathcal{B} \right| + \left| \mathcal{Q} \right|} \right)} \right]
\mathop {\sup }\limits_{0 \leqslant s \leqslant t} {\left\| {{\phi _1}\left( s \right) - {\phi _2}\left( s \right)} \right\|_{{L^1}}}
\end{flalign*}
Now we let $l \in C\left( {\left[ {0,{T_\phi }} \right);L_ + ^1} \right)$ be the positive solution of \eqref{ADP}, then for any $t \in \left[ {0,{T_\phi }} \right)$, $(H.5)$ is easy to verify:
$$\left( {F\left( l \right)} \right)\left( t \right) + \int_0^\emph{M} {G\left( {l\left( t \right)} \right)\left( a \right)da}  \leqslant {S_0}\left|\mathcal{B} \right|\int_0^\emph{M} {l\left( t \right)\left( a \right)da} $$
So by Theorem \ref{global solution}, there is a positive global solution of \eqref{ADP}.
\end{proof}

%%%%%%%%%%%%%%%%%%%%%%%%%%%%%%%%%%%%%%%%%%%%%%%%%%%%%%%%%%%%%%%%%%%%%%%%%%%%%%%%%%%%%%%%%%%%%%%%%%%%%%%%%%%%%%%%%%%%%%%%%%%%%%%%%%%%%%%%%%%%%%%%%%%%%%%%%%%%%%%%%%%%%%%%%%%%%%%%%%%%%%%%%%%%%%%%%%%%%%%%%%%%%%%%%%%%%%%%%%%%%%%%%%%%%%%%%%%%%%%%%%%%%%%%%%%%%%%%
\begin{proof}[Proof of Proposition \ref{existence of related problem}]
Proof of the existence of the global positive solution is based on the following four lemmas. An easy computation can be done to testify that a solution of \eqref{u formula} is also a solution in the sense of \eqref{ADP} with aging function $\mathcal{P}$ and birth function $\mathcal{H}$.
\end{proof}
%%%%%%%%%%%%%%%%%%%%%%%%%%%%%%%%%%%%%%%%%%%%%%%%%%%%%%%%%%%%%%%%%%%%%%%%%%%%%%%%%%%%%%%%%%%%%%%%%%%%%%%%%%%%%%%%%%%%%%%%%%%%%%%%%%%%%%%%%%%%%%%%%%%%%%%%%%%%%%%%%%%%%%%%%%%%%%%%%%%%%%%%%%%%%%%%%%%%%%%%%%%%%%%%%%%%%%%%%%%%%%%%%%%%%%%%%%%%%%%%%%%%%%%%%%%%%%%%
\begin{lem}\label{u lemma 1}
Let the assumptions in Proposition \ref{existence of related problem} hold, and let $r>0$. There exists $T>0$ such that if ${\left\| \phi  \right\|_{{L^1}}} \leqslant r$, then there is a unique function $u \in {C_{T, + }}$ such that $u$ is a solution of \eqref{u formula} on $\left[ {0,T} \right]$.
\end{lem}
\begin{proof}
We can choose a sufficiently small $T>0$ and define 
$$S: = \left\{ {u \in {C_{T, + }}:u\left( t \right) = \phi ,{{\left\| u \right\|}_{{C_T}}} \leqslant 2r} \right\}$$. 
An argument which is similar to that of Proposition \ref{existence and uniqueness of integral ADP solution} can be used to show that a mapping defined by \eqref{u formula} from $S$ into $S$ is a strict contraction. Thus, the unique fixed point is a positive solution of \eqref{u formula} in ${C_{T, + }}$.
\end{proof}

%%%%%%%%%%%%%%%%%%%%%%%%%%%%%%%%%%%%%%%%%%%%%%%%%%%%%%%%%%%%%%%%%%%%%%%%%%%%%%%%%%%%%%%%%%%%%%%%%%%%%%%%%%%%%%%%%%%%%%%%%%%%%%%%%%%%%%%%%%%%%%%%%%%%%%%%%%%%%%%%%%%%%%%%%%%%%%%%%%%%%%%%%%%%%%%%%%%%%%%%%%%%%%%%%%%%%%%%%%%%%%%%%%%%%%%%%%%%%%%%%%%%%%%%%%%%%%%%
\begin{lem}\label{u lemma 2}
Let the assumptions in Proposition \ref{existence of related problem} hold, let $T>0$, and let $u \in {C_{T, + }}$ such that $u$ is a solution of \eqref{u formula} on $\left[ {0,T} \right]$. Let $\hat T>0$ and let $\hat u \in {C_{T+{\hat T}, + }}$ such that $\hat u\left( t \right) = u\left( t \right)$ for $t \in \left[ {0,T} \right]$, and for $t \in \left( {T,T + \hat T} \right]$, $\hat u$ satisfies the integral equation:
\begin{equation*}
  \hat u\!\left( t \right)\!\left( a \right) = \left\{ \begin{matrix} 
  \left( {\mathcal{H}\left( \hat u \right)} \right)\!\left( {t - a} \right){e^{ - \int_0^a {\mu \left( b \right)db} }},\;0< a < t-T \hfill \cr 
  u\left( T \right)\!\left( {a - t + T} \right){e^{ - \int_{a - t + T}^a {\mu \left( b \right)db} }},\;t-T \leq a \leq \emph{M}\hfill \cr 
 \end{matrix}  \right.
\end{equation*}
where $\mathcal{H}$ is as stated in Proposition \ref{existence of related problem}. Then $\hat u$ is a solution of \eqref{u formula} on $\left[ {0,T + \hat T} \right]$.
\end{lem}
\begin{proof}
The proof is similar to that of Proposition \ref{semigroup property}.
\end{proof}
%%%%%%%%%%%%%%%%%%%%%%%%%%%%%%%%%%%%%%%%%%%%%%%%%%%%%%%%%%%%%%%%%%%%%%%%%%%%%%%%%%%%%%%%%%%%%%%%%%%%%%%%%%%%%%%%%%%%%%%%%%%%%%%%%%%%%%%%%%%%%%%%%%%%%%%%%%%%%%%%%%%%%%%%%%%%%%%%%%%%%%%%%%%%%%%%%%%%%%%%%%%%%%%%%%%%%%%%%%%%%%%%%%%%%%%%%%%%%%%%%%%%%%%%%%%%%%%%
\begin{lem}\label{u lemma 3}
Let the assumptions in Proposition \ref{existence of related problem} hold, let $\phi  \in L_ + ^1$, and let $u$ be the solution of \eqref{u formula} on its maximal interval of existence $\left[ {0,{T_\phi }} \right)$. If ${T_\phi } < \infty $, then $\mathop {\lim \sup }\limits_{t \to {T_\phi }^ - } {\left\| {u\left( t \right)} \right\|_{{L^1}}} = \infty $.
\end{lem}
\begin{proof}
The proof is similar to that of Theorem $2.3$ in (\citealt{webbbook}).
\end{proof}

%%%%%%%%%%%%%%%%%%%%%%%%%%%%%%%%%%%%%%%%%%%%%%%%%%%%%%%%%%%%%%%%%%%%%%%%%%%%%%%%%%%%%%%%%%%%%%%%%%%%%%%%%%%%%%%%%%%%%%%%%%%%%%%%%%%%%%%%%%%%%%%%%%%%%%%%%%%%%%%%%%%%%%%%%%%%%%%%%%%%%%%%%%%%%%%%%%%%%%%%%%%%%%%%%%%%%%%%%%%%%%%%%%%%%%%%%%%%%%%%%%%%%%%%%%%%%%%%%
\begin{lem}\label{u lemma 4}
Let the assumptions in Proposition \ref{existence of related problem} hold, let $\phi  \in L_ + ^1$, and let $u$ be the positive solution of \eqref{u formula} on its maximal interval of existence $\left[ {0,{T_\phi }} \right)$. Then $\exists \omega  \in \mathbb{R}$, and for $t \in \left[ {0,{T_\phi }} \right)$, 
${\left\| {u\left( t \right)} \right\|_{{L^1}}} \leqslant {\left\| \phi  \right\|_{{L^1}}}{e^{\omega t}}$.
So ${T_\phi } = \infty $, there is a global positive solution of \eqref{u formula}.
\end{lem}
\begin{proof}
For $t \in \left[ {0,{T_\phi }} \right)$, we estimate as follows (here we assume $t \leqslant \emph{M}$, then $t>\emph{M}$ leads to a simpler case):
\begin{flalign*}
&{\left\| {u\left( t \right)} \right\|_{{L^1}}} = \int_0^\emph{M} {u\left( t \right)\!\left( a \right)da} 
\leqslant  \int_0^t {\left( {\mathcal{H}\left( u \right)} \right)\left( {t - a} \right)da}  + \int_t^\emph{M} {\phi \left( {a - t} \right)da} \\
\leqslant & \int_0^t {{S_0}\left| \mathcal{B} \right|{{\left\| {u\left( {t - a} \right)} \right\|}_{{L^1}}}da}  + \int_t^\emph{M} {\phi \left( {a - t} \right)da} \\
\leqslant & {S_0}\left| \mathcal{B} \right|\int_0^t {{{\left\| {u\left( s \right)} \right\|}_{{L^1}}}ds}  + \int_0^{\emph{M} - t} {\phi \left( s \right)ds}\\
\leqslant & {S_0}\left| \mathcal{B} \right|\int_0^t {{{\left\| {u\left( s \right)} \right\|}_{{L^1}}}ds}  + {\left\| \phi  \right\|_{{L^1}}}
\end{flalign*}
Then by Gronwall's Inequality, we have ${\left\| {u\left( t \right)} \right\|_{{L^1}}} \leqslant {\left\| \phi  \right\|_{{L^1}}}{e^{{S_0}\left| \mathcal{B} \right|t}}$. Then by Lemma \ref{u lemma 3}, ${T_\phi } = \infty $, hence there is a positive global solution to \eqref{u formula}.
\end{proof}
%%%%%%%%%%%%%%%%%%%%%%%%%%%%%%%%%%%%%%%%%%%%%%%%%%%%%%%%%%%%%%%%%%%%%%%%%%%%%%%%%%%%%%%%%%%%%%%%%%%%%%%%%%%%%%%%%%%%%%%%%%
%%%%%%%%%%%%%%%%%%%%%%%%%%%%%%%%%%%%%%%%%%%%%%%%%%%%%%%%%%%%%%%%%%%%%%%%%%%%%%%%%%%%%%%%%%%%%%%%%%%%%%%%%%%%%%%%%%%%%%%%%%

%%%%%%%%%%%%%%%%%%%%%%%%%%%%%%%%%%%%%%%%%%%%%%%%%%%%%%%%%%%%%%%%%%%%%%%%%%%%%%%%%%%%%%%%%%%%%%%%%%%%%%%%%%%%%%%%%%%%%%%%%%
%%%%%%%%%%%%%%%%%%%%%%%%%%%%%%%%%%%%%%%%%%%%%%%%%%%%%%%%%%%%%%%%%%%%%%%%%%%%%%%%%%%%%%%%%%%%%%%%%%%%%%%%%%%%%%%%%%%%%%%%%%

\begin{proof}[Proof of Theorem \ref{solution in form of u}]
Let $T>0$, we assume that $u \in C\left( {\left[ {0,T} \right];L_ + ^1} \right)$ satisfies \eqref{u formula} for $t \in \left[ {0,T} \right]$, then $u$ satisfies the following conditions:
\begin{flalign}\label{ADP of u}
\begin{split}
&\mathop {\lim }\limits_{h \to {0^ + }} \int_0^\emph{M} {\left| {{h^{ - 1}}\left[ {u\left( {t + h} \right)\left( {a + h} \right) - u\left( t \right)\left( a \right)} \right] + \mu \left( a \right)u\left( t \right)\left( a \right)} \right|da}  = 0\\
&\mathop {\lim }\limits_{h \to {0^ + }} {h^{ - 1}}\int_0^h {\left| {u\left( {t + h} \right)\left( a \right) - \left( {\mathcal{H}\left( u \right)} \right)\left( t \right)} \right|da = 0}\\
&u\left( 0 \right) = \phi
\end{split}
\end{flalign}
We will show that $l \in C\left( {\left[ {0,T} \right];L_ + ^1} \right)$ as obtained from \eqref{solution formula}, satisfies \eqref{ADP} with the aging function $G$ in (P.1) and the birth function $F$ in (P.2). For the first condition in \eqref{ADP}, we have the following estimation:
\begin{flalign*}
&{h^{ - 1}}\left[ {l\left( {t + h} \right)\left( {a + h} \right) - l\left( t \right)\left( a \right)} \right] + \mu \left( a \right)l\left( t \right)\left( a \right) + \mathcal{T}\left( {l\left( t \right)} \right)l\left( t \right)\left( a \right)\\
= & \frac{{{h^{ - 1}}\left[ {u\left( {t + h} \right)\left( {a + h} \right) - u\left( t \right)\left( a \right)} \right] + \mu \left( a \right)u\left( t \right)\left( a \right)}}{{1 + \int_0^{t + h} {\mathcal{T}\left( {u\left( s \right)} \right)ds} }}\\
+ & \left[ {\frac{{\mu \left( a \right)u\left( t \right)\left( a \right)}}{{1 + \int_0^t {\mathcal{T}\left( {u\left( s \right)} \right)ds} }} - \frac{{\mu \left( a \right)u\left( t \right)\left( a \right)}}{{1 + \int_0^{t + h} {\mathcal{T}\left( {u\left( s \right)} \right)ds} }}} \right]\\
+ &  {{h^{ - 1}}\left[ {\frac{{u\left( t \right)\left( a \right)}}{{1 + \int_0^{t + h} {\mathcal{T}\left( {u\left( s \right)} \right)ds} }} - \frac{{u\left( t \right)\left( a \right)}}{{1 + \int_0^t {\mathcal{T}\left( {u\left( s \right)} \right)ds} }}} \right] + \frac{{\mathcal{T}\left( {u\left( t \right)} \right)u\left( t \right)\left( a \right)}}{{{{\left( {1 + \int_0^t {\mathcal{T}\left( {u\left( s \right)} \right)ds} } \right)}^2}}}} \\
: = & {I_1} + {I_2} + {I_3}
\end{flalign*}
By \eqref{ADP of u}, ${I_1} \to 0$ as $h \to 0$. ${I_2} \to 0$ as $h \to 0$ because of the absolute continuity of Lebesgue integral. If we compute the derivative of function $f\left( t \right): = \frac{1}{{1 + \int_0^t {\mathcal{T}\left( {u\left( s \right)} \right)ds} }}$, we get ${I_3} \to 0$ as $h \to 0$. Hence the first limit in the solution definition \eqref{ADP} is satisfied. For the second condition in \eqref{ADP}, we have
\begin{flalign*}
&\int_0^h {\left| {l\left( {t + h} \right)\left( a \right) - \left( {F\left( l \right)} \right)\left( t \right)} \right|} da\\
= & \int_0^h {\left| {l\left( {t + h} \right)\left( a \right) - {S_0}\mathcal{B}\left( {l\left( t \right)} \right){e^{ - \int_0^t {\mathcal{B}\left( {l\left( s \right)} \right) + \mathcal{Q}\left( {l\left( s \right)} \right)ds} }}} \right|da} \\
= & \int_0^h {\left| {\frac{{u\left( {t + h} \right)\left( a \right) - \left( {\mathcal{H}\left( u \right)} \right)\left( t \right)}}{{1 + \int_0^t {\mathcal{T}\left( {u\left( s \right)} \right)ds} }}} \right|} da \to 0,\,\,\left( {h \to {0^ + }} \right)
\end{flalign*}
The third condition in \eqref{ADP} is straightforward. Then by Proposition \ref{existence of related problem}, we can find a $u \in C\left( {\left[ {0,\infty } \right);L_ + ^1} \right)$ that satisfies \eqref{u formula}. Then a positive global solution to problem \eqref{ADP} can be obtained by \eqref{solution formula}, which is exactly the unique positive global solution to \eqref{ADP}.
\end{proof}
%%%%%%%%%%%%%%%%%%%%%%%%%%%%%%%%%%%%%%%%%%%%%%%%%%%%%%%%%%%%%%%%%%%%%%%%%%%%%%%%%%%%%%%%%%%%%%%%%%%%%%%%%%%%%%%%%%%%%%%%%%
%%%%%%%%%%%%%%%%%%%%%%%%%%%%%%%%%%%%%%%%%%%%%%%%%%%%%%%%%%%%%%%%%%%%%%%%%%%%%%%%%%%%%%%%%%%%%%%%%%%%%%%%%%%%%%%%%%%%%%%%%%
\begin{proof}[Proof of Theorem \ref{asymptotic}]
Let $u \in C\left( {\left[ {0,\infty } \right);L_ + ^1} \right)$ be the solution to \eqref{u formula}, and let $i \in C\left( {\left[ {0,\infty } \right);L_ + ^1} \right)$ as defined in \eqref{solution formula}, which is a solution to problem \eqref{whole model} in the sense of \eqref{ADP}. For convenience in this proof, we use the notation $i\left( {a,t} \right): = i\left( t \right)\!\left( a \right)$ and $u\left( {a,t} \right): = u\left( t \right)\!\left( a \right)$, for $a \in \left[ {0,\emph{M}} \right],\;t \in \left[ {0,\infty } \right)$. Firstly, by \eqref{whole model} we have
$$S\left( t \right) = {S_0}{e^{ - \int_0^t {\mathcal{B}\left( {i\left( { \cdot ,s} \right)} \right) + \mathcal{Q}\left( {i\left( { \cdot ,s} \right)} \right)ds} }}$$
which is a positive non-increasing continuous function of $t \in \left[ {0,\infty } \right)$. So $\mathop {\lim }\limits_{t \to \infty } S\left( t \right)$ exists, and we denote it as $\mathop {\lim }\limits_{t \to \infty } S\left( t \right) = {S_\infty } \geqslant 0$. Next we estimate the following:
\begin{flalign*}
&\int_0^\infty  {\mathcal{B}\left( {i\left( { \cdot ,t} \right)} \right) + \mathcal{Q}\left( {i\left( { \cdot ,t} \right)} \right)dt}  \leqslant \int_0^\infty  {\left( {\left| \mathcal{B} \right| + \left| \mathcal{Q} \right|} \right){{\left\| {i\left( { \cdot ,t} \right)} \right\|}_{{L^1}}}dt} \\
= & \left( {\left| \mathcal{B} \right| + \left| \mathcal{Q} \right|} \right)\int_0^\infty  {\int_0^\emph{M} {\frac{{u\left( {a,t} \right)}}{{1 + \int_0^t {\mathcal{T}\left( {u\left( { \cdot ,s} \right)} \right)ds} }}dadt} } \\
= & \left( {\left| \mathcal{B} \right| + \left| \mathcal{Q} \right|} \right)\int_0^\emph{M} {\int_0^a {\frac{{\phi \left( {a - t} \right){e^{ - \int_{a - t}^a {\mu \left( b \right)db} }}}}{{1 + \int_0^t {\mathcal{T}\left( {u\left( { \cdot ,s} \right)} \right)ds} }}dt} da} \\
& + \left( {\left| \mathcal{B} \right| + \left| \mathcal{Q} \right|} \right)\int_0^\emph{M} {\int_a^\infty  {\frac{{u\left( {0,t - a} \right){e^{ - \int_0^a {\mu \left( b \right)db} }}}}{{1 + \int_0^t {\mathcal{T}\left( {u\left( { \cdot ,s} \right)} \right)ds} }}dt} da} \\
: = & \left( {\left| \mathcal{B} \right| + \left| \mathcal{Q} \right|} \right)\left( {{I_1} + {I_2}} \right)
\end{flalign*}
Estimate ${I_1}$ and ${I_2}$ separately:
\begin{flalign*}
&{I_1} \leqslant \int_0^\emph{M} {\int_0^a {{\phi}\left( {a - t} \right){e^{ - \int_{a - t}^a {\mu \left( b \right)db} }}dtda} } 
=  \int_0^\emph{M} {\int_0^a {{\phi}\left( s \right){e^{ - \int_s^a {\mu \left( b \right)db} }}dsda} } \\
= & \int_0^{{a_0}} {\int_0^a {{\phi}\left( s \right){e^{ - \int_s^a {\mu \left( b \right)db} }}dsda} }  + \int_{{a_0}}^\emph{M} {\int_0^a {{\phi}\left( s \right){e^{ - \int_s^a {\mu \left( b \right)db} }}dsda} } \\
\leqslant & \int_0^{{a_0}} {\int_0^a {{\phi}\left( s \right)dsda} }  + \int_{{a_0}}^\emph{M} {\int_0^a {{\phi}\left( s \right){e^{ - {\mu _0} \left( {a - s} \right)}}dsda} } \\
\leqslant & {a_0}{\left\| {\phi} \right\|_{{L^1}}} + \int_0^\emph{M} {\int_0^a {{\phi}\left( s \right){e^{ - {\mu_0} \left( {a - s} \right)}}dsda} } \\
= & {a_0}{\left\| {{\phi}} \right\|_{{L^1}}} + \int_0^\emph{M} {\int_s^\emph{M} {{\phi}\left( s \right){e^{ - {\mu_0} \left( {a - s} \right)}}dads} } \\
= & {a_0}{\left\| {{\phi}} \right\|_{{L^1}}} + \int_0^\emph{M} {\int_0^{\emph{M} - s} {{\phi}\left( s \right){e^{ - {\mu_0} \tau }}d\tau ds} } \\
\leqslant & {a_0}{\left\| {{\phi}} \right\|_{{L^1}}} + \int_0^\emph{M} {\int_0^\emph{M} {{\phi}\left( s \right){e^{ - {\mu_0} \tau }}d\tau ds} } \\
= & {a_0}{\left\| {{\phi}} \right\|_{{L^1}}} + \left( {\int_0^\emph{M} {{e^{ - {\mu_0} \tau }}d\tau } } \right){\left\| {{\phi}} \right\|_{{L^1}}} < \infty 
\end{flalign*}
\begin{flalign*}
&{I_2} = \int_0^\emph{M} {\int_a^\infty  {\frac{{u\left( {0,t - a} \right){e^{ - \int_0^a {\mu \left( b \right)db} }}}}{{1 + \int_0^t {\mathcal{T}\left( {u\left( { \cdot ,s} \right)} \right)ds} }}dtda} } \\
\leqslant & \int_0^\emph{M} {\int_0^\infty  {\frac{{u\left( {0,\tau } \right){e^{ - \int_0^a {\mu \left( b \right)db} }}}}{{1 + \int_0^{a + \tau } {\mathcal{T}\left( {u\left( { \cdot ,s} \right)} \right)ds} }}d{\tau}da} } \\
\leqslant & \int_0^\emph{M} {\int_0^\infty  {\frac{{u\left( {0,\tau } \right){e^{ - \int_0^a {\mu \left( b \right)db} }}}}{{1 + \int_0^\tau  {\mathcal{T}\left( {u\left( { \cdot ,s} \right)} \right)ds} }}d{\tau}da} } \\
 = & \left( {\int_0^\infty  {i\left( {0,\tau } \right)d\tau } } \right)\left( {\int_0^\emph{M} {{e^{ - \int_0^a {\mu \left( b \right)db} }}} da} \right)\\
\leqslant & \left( {\int_0^\infty  {i\left( {0,\tau } \right)d\tau } } \right)\left( {{a_0} + \int_{{a_0}}^\emph{M} {{e^{ - a{\mu _0}}}da} } \right)
\end{flalign*}
With the assumption on function $\mu$, we can find constants ${C_1},\,{C_2} > 0$ such that:
$$\int_0^\infty  {\mathcal{B}\left( {i\left( { \cdot ,t} \right)} \right) + \mathcal{Q}\left( {i\left( { \cdot ,t} \right)} \right)dt}  \leqslant {C_1} + {C_2}\int_0^\infty  {i\left( {0,t} \right)dt} $$
Since $i$ is obtained from \eqref{solution formula}, we have:
$$i\left( {0,t} \right) = \frac{{u\left( {0,t} \right)}}{{1 + \int_0^t {\mathcal{T}\left( {u\left( { \cdot ,s} \right)} \right)ds} }} = S\left( t \right)\mathcal{B}\left( {i\left( { \cdot ,t} \right)} \right)$$
Then by the differential equation of $S(t)$ in \eqref{whole model}, we have: 
$$i\left( {0,t} \right) =  - \frac{{dS\left( t \right)}}{{dt}} - {\mathcal{Q}\left( {i\left( { \cdot ,t} \right)} \right)}S\left( t \right) \leqslant  - \frac{{dS\left( t \right)}}{{dt}}$$
Integrate on both sides with respect to $t$ of the above inequality, 
$$\int_0^\infty  {i\left( {0,t} \right)dt}  \leqslant {S_0} - {S_\infty } < \infty $$
Hence one of the conclusion is proved:
$$\mathop {\lim }\limits_{t \to \infty } S\left( t \right) = {S_0}{e^{ - \int_0^\infty  {\mathcal{B}\left( {i\left( { \cdot ,s} \right)} \right) + \mathcal{Q}\left( {i\left( { \cdot ,s} \right)} \right)ds} }} \geqslant {S_0}{e^{ - {C_1} - {C_2}\int_0^\infty  {i\left( {0,t} \right)dt} }} > 0$$
Moreover, it can be derived from the differential equation system \eqref{whole model} that:
\begin{flalign}\label{consistence}
\begin{split}
&S\left( t \right) + I\left( t \right) + \int_0^t {\mathcal{Q}\left( {i\left( { \cdot ,s} \right)} \right)I\left( s \right)ds}  + \int_0^t {\mathcal{T}\left( {i\left( { \cdot ,s} \right)} \right)I\left( s \right)ds}  + \int_0^t {i\left( {\emph{M},s} \right)ds} \\
 + & \int_0^t {\int_0^\emph{M} {\mu \left( a \right)i\left( {a,s} \right)dads} }  = {S_0} + I\left( 0 \right)
\end{split}
\end{flalign}
where the four integrals in \eqref{consistence} are non-decreasing with respect to the variable $t$ and have ${S_0} + I\left( 0 \right)$ as an upper bound. So the four integrals all have finite limit as ${t \to \infty }$. Then the fact that $\mathop {\lim }\limits_{t \to \infty } S\left( t \right)$ exists implies that $\mathop {\lim }\limits_{t \to \infty } I\left( t \right)$ exists. We can estimate $\int_0^\infty  {I\left( t \right)dt}  = \int_0^\infty  {{{\left\| {i\left( t \right)} \right\|}_{{L^1}}}dt} $ similarly as we did in the beginning of this proof and get $\int_0^\infty  {I\left( t \right)dt}  < \infty $, which implies the conclusion $\mathop {\lim }\limits_{t \to \infty } I\left( t \right) = 0$.

\end{proof}

%\section{Section title}
%\label{sec:1}
%Text with citations \citealt{RefB} and \citealt{RefJ}.
%\subsection{Subsection title}
%\label{sec:2}
%as required. Don't forget to give each section
%and subsection a unique label (see Sect.~\ref{sec:1}).
%\paragraph{Paragraph headings} Use paragraph headings as needed.
%\begin{equation}
%\end{equation}

%\begin{figure}
%  \includegraphics{example.eps}
%\end{figure}
%\begin{figure*}
%  \includegraphics[width=0.75\textwidth]{example.eps}
%\end{figure*}
%\begin{table}
%\begin{tabular}{lll}
%\hline\noalign{\smallskip}
%\end{tabular}
%\end{table}

\end{document}